\documentclass[a4paper, 12pt]{article}

\usepackage[sort&compress]{natbib}
\bibpunct{(}{)}{;}{a}{}{,} 

\usepackage{amsthm, amsmath, amssymb, mathrsfs, multirow, url, subfigure}
\usepackage{graphicx} 
\usepackage{ifthen} 
\usepackage{amsfonts}
\usepackage[usenames]{color}
\usepackage{fullpage}
\usepackage{tikz}
\usetikzlibrary{shapes.symbols, hobby, positioning}
\usepackage[ruled,vlined]{algorithm2e}

\theoremstyle{plain} 
\newtheorem{thm}{Theorem}
\newtheorem{cor}{Corollary}

\theoremstyle{definition}

\theoremstyle{remark}
\newtheorem{ex}{Example}
\newtheorem{remark}{Remark}

\newtheorem*{ex1}{Example {\em 1}}
\newtheorem*{ex3}{Example {\em 3}}
\newtheorem*{illus}{Illustration}

\newcommand{\prob}{\mathsf{P}}

\newcommand{\unif}{{\sf Unif}}
\newcommand{\nm}{{\sf N}}

\newcommand{\chisq}{{\sf ChiSq}}

\newcommand{\RR}{\mathbb{R}}

\newcommand{\ZZ}{\mathbb{Z}}

\newcommand{\TT}{\mathbb{T}}

\newcommand{\eps}{\varepsilon}

\newcommand{\prior}{\mathsf{Q}}

\newcommand{\kernel}{\mathsf{K}}
\newcommand{\marg}{\mathsf{M}}

\newcommand{\cred}{\mathscr{C}}

\newcommand{\lPi}{\amalg}
\newcommand{\uPi}{\Pi}

\newcommand{\inn}{\text{inn}}
\newcommand{\out}{\text{out}}

\title{An efficient Monte Carlo method for valid prior-free possibilistic statistical inference}
\author{Ryan Martin\footnote{Department of Statistics, North Carolina State University, {\tt rgmarti3@ncsu.edu}}
}
\date{\today}

\begin{document}

\maketitle 

\begin{abstract}
Inferential models (IMs) offer prior-free, Bayesian-like posterior degrees of belief designed for statistical inference, which feature a frequentist-like calibration property that ensures reliability of said inferences.  The catch is that IMs' degrees of belief are possibilistic rather than probabilistic and, since the familiar Monte Carlo methods approximate probabilistic quantities, there are significant computational challenges associated with putting this framework into practice.  The present paper overcomes these challenges by developing a new Monte Carlo method designed specifically to approximate the IM's possibilistic output.  The proposal is based on a characterization of the possibilistic IM's credal set, which identifies the ``best probabilistic approximation'' of the IM as a mixture distribution that can be readily approximated and sampled from.  These samples can then be transformed into an approximation of the possibilistic IM.  Numerical results are presented highlighting the proposed approximation's accuracy and computational efficiency.

\smallskip

\emph{Keywords and phrases:} confidence distribution; credal set; Gaussian possibility; inferential model; inner probabilistic approximation.
\end{abstract}

\section{Introduction}
\label{S:intro}

Numerous efforts to advance Fisher's vision of fiducial inference have been made over the years.  Among these efforts, arguably ``one of the original statistical innovations of the 2010s'' \citep{cui.hannig.im} is the {\em inferential model} (IM) framework put forward in \citet{imbasics, imcond, immarg} and later synthesized in \citet{imbook}.  What's unique about the IM framework is that its output takes the form of an imprecise probability---specifically, a possibility measure---and, consequently, inference is based on possibilistic reasoning: hypotheses assigned small possibility are refuted by the data while hypotheses whose complements are assigned small possibility are corroborated by the data.  This shift from probability to possibility theory has principled motivations: the IM's possibilities satisfy a frequentist-style {\em validity} property that Bayesian/fiducial probabilities do not satisfy, i.e., the IM's possibility assigned to true hypotheses tends to be not small, hence it's a provably rare event that the data-driven IM output refutes a true hypothesis or corroborates a false hypothesis.  More details about IMs and their properties are provided in Section~\ref{S:background} below; some basic background on possibility theory is presented in Appendix~\ref{A:primer}.  Despite all that the IM framework has to offer, it has been slow to gain traction, largely due to computational challenges.  Indeed, as shown recently in \citet{gong.jasa.mult}, the familiar Monte Carlo methods used to approximate ordinary probabilities aren't enough to approximate imprecise probabilities; something more is needed.  The goal of this paper is to identify the aforementioned ``something more'' for possibilistic IMs and to develop the corresponding computational machinery needed to make the IM framework accessible for everyday use in applications.  

The jumping off point here is the fact that all (coherent) imprecise probabilities, including the possibilistic IM's output, correspond to a (non-empty) set of ordinary or precise probabilities, called the {\em credal set}.  Different brands of imprecise probabilities have their own mathematical properties, and one way this distinction manifests is in the kinds of constraints imposed on their associated credal sets.  Possibility measures are among the simplest imprecise probabilities and, in turn, their credal sets have a relatively simple characterization.  In particular, results like that in Theorem~\ref{thm:credal.char} below make it relatively easy to identify probability distributions in the credal set that ``best approximate'' the IM's possibilistic output.  The best approximation is relatively simple mathematically but might be difficult to compute; fortunately, further approximations can be made.  Such second-level approximations can take various forms but, inspired by the asymptotic results in \citet{imbvm.ext}, here I'll use a family of Gaussian approximations as proposed in \citet{imvar.ext}.  These individual Gaussian approximations on their own are not fully satisfactory, as Cella and Martin acknowledge, because they only ensure accuracy in a limited sense.  If the individual approximations could be appropriately stitched together, then that would give an overall accurate approximation---but Cella and Martin weren't able to provide a suitable stitching procedure.  Here, using Theorem~\ref{thm:credal.char}, I propose a stitching procedure that achieves the relevant statistical and computational desiderata, thereby closing the critical gap left by Cella and Martin and facilitating a principled probabilistic approximation to the IM's output.  I also show how this probabilistic approximation can be transformed back into a possibility measure that is both easy to compute and accurately approximates the target IM output.  

The remainder of this paper is organized as follows.  Section~\ref{S:background} provides the necessary background on possibilistic IMs, highlighting the existing computational bottlenecks.  The credal set characterization and proposed Monte Carlo strategy are presented in Section~\ref{S:mcim}.  Three illustrative examples, including logistic regression and a semiparametric problem involving censored data, are presented in Section~\ref{S:examples}.  These are cases where direct-but-naive IM computations are expensive but not out of reach, so the accuracy of the proposed approximation can be assessed.  Section~\ref{S:simulations} shows empirically that the proposed IM approximation performs well compared to standard Bayesian and likelihood-based solutions.  A higher-dimensional logistic regression example is considered in Section~\ref{S:application}, where the proposed IM approximation can, in a matter of seconds, easily solve a problem that was practically out of reach by previous efforts in the IM literature.  
The insights and results here in the present paper make up an important first step in a series of developments needed to create an IM toolbox for practitioners.  The concluding remarks in Section~\ref{S:discuss} describe, among other things, my vision of where these developments can and need to go.  Relevant technical details are given in the appendix/supplement.

\section{Background}
\label{S:background}

The first IM developments relied on random sets and their corresponding belief functions.  More recent developments in \citet{martin.partial2}, building on \citet{plausfn, gim}, define a possibilistic IM by applying a version of the  probability-to-possibility transform to the model's relative likelihood.  This shift is philosophically important, but the present review focuses on the details, properties, and computation of this possibilistic IM. 

Consider a parametric model $\{ \prob_\theta: \theta \in \TT\}$ consisting of probability distribution supported on a sample space $\ZZ$, indexed by a parameter space $\TT \subseteq \RR^d$.  Suppose that the observable data $Z$, taking values in $\ZZ$, is a sample from the distribution $\prob_\Theta$, where $\Theta \in \TT$ is the unknown/uncertain ``true value.''  The model and observed data $Z=z$ together determine a likelihood function $\theta \mapsto L_z(\theta)$ and a corresponding relative likelihood
\[ 
R(z,\theta) = \frac{L_{z}(\theta)}{\sup_\vartheta L_{z}(\vartheta)}, \quad z \in \ZZ, \quad \theta \in \TT.
\]
I'll implicitly assume that the denominator is finite for almost all $z$.  As is typical in the literature, I will also assume that prior information about $\Theta$ is vacuous. 

The relative likelihood itself defines a data-dependent possibility contour, i.e., a non-negative function such that $\sup_\theta R(z,\theta) = 1$ for each $z$; see Appendix~\ref{A:primer}.  This contour, in turn, determines a possibility measure that can be used for uncertainty quantification about $\Theta$, given $Z=z$, which has been extensively studied in the literature \citep[e.g.,][]{shafer1982, wasserman1990b, denoeux2006, denoeux2014}.  In particular, these references suggest assigning data-driven possibility values to hypotheses $H$ about $\Theta$ via the rule 
\[ H \mapsto \sup_{\theta \in H} R(z, \theta), \quad H \subseteq \TT. \]
This purely likelihood-driven possibility has a number of desirable properties, which I'll not get into here.  What it lacks, however, is a justification for why the ``possibilities'' assigned to hypotheses about $\Theta$ have belief-forming inferential weight.  With vacuous prior information, there's no Bayesian justification behind these possibility assignments, so justification can only come from a frequentist-like calibration property---i.e., assigning small possibilities to true hypotheses is a provably rare event---but the purely likelihood-based possibility assignment doesn't meet this requirement in a practically useful way.  So, while the relative likelihood provides a natural, data-driven parameter ranking in terms of model fit, this is insufficient for principled and reliable statistical inference.  


Fortunately, however, it is conceptually straightforward to achieve the desired calibration by applying what \citet{martin.partial} calls ``validification''---a version of the probability-to-possibility transform \citep[e.g.,][]{dubois.etal.2004, hose2022thesis}.  In particular, for observed data $Z=z$, the possibilistic IM's contour is defined as
\begin{equation}
\label{eq:contour}
\pi_{z}(\theta) = \prob_\theta\bigl\{ R(Z,\theta) \leq R(z, \theta) \bigr\}, \quad \theta \in \TT,
\end{equation}
and the possibility measure---or upper probability---is likewise defined as 
\begin{equation}
\label{eq:maxitive}
\uPi_{z}(H) = \sup_{\theta \in H} \pi_{z}(\theta), \quad H \subseteq \TT.
\end{equation}
It won't be needed in the present paper, but there's a corresponding necessity measure, or lower probability, defined via conjugacy: $\lPi_{z}(H) := 1 - \uPi_{z}(H^c)$.  An essential feature of this IM construction is its so-called {\em validity property}:
\begin{equation}
\label{eq:valid}
\sup_{\theta \in \TT} \prob_\theta\bigl\{ \pi_{Z}(\theta) \leq \alpha \bigr\} \leq \alpha, \quad \text{for all $\alpha \in [0,1]$}. 
\end{equation}
This has a number of important consequences.  First, \eqref{eq:valid} immediately implies that 
\begin{equation}
\label{eq:conf.set}
C_\alpha(z) = \{\theta \in \TT: \pi_{z}(\theta) \geq \alpha\}, \quad \alpha \in [0,1]
\end{equation}
is a $100(1-\alpha)$\% frequentist confidence set, i.e., $\sup_{\theta \in \TT} \prob_\theta\bigl\{ C_\alpha(Z) \not\ni \theta \bigr\} \leq \alpha$.  Second, from \eqref{eq:maxitive} and \eqref{eq:valid}, it readily follows that 
\begin{equation}
\label{eq:valid.alt}
\sup_{\theta \in H} \prob_\theta\bigl\{ \uPi_{Z}(H) \leq \alpha \bigr\} \leq \alpha, \quad \text{all $\alpha \in [0,1]$, all $H \subseteq \TT$}. 
\end{equation}
In words, a valid IM assigns possibility $\leq \alpha$ to true hypotheses at rate $\leq \alpha$ as a function of data $Z$.  This gives the IM its ``inferential weight''---\eqref{eq:valid.alt} implies that $\uPi_{z}(H)$ is not expected to be small when $H$ is true, so one is inclined to doubt the truthfulness of a hypothesis $H$ if $\uPi_{z}(H)$ is small.  Third, the above property ensures that the possibilistic IM is safe from false confidence \citep{balch.martin.ferson.2017, martin.nonadditive, martin.belief2024}, unlike all default-prior Bayes and fiducial solutions.  An even stronger, {\em uniform} version of \eqref{eq:valid.alt} holds, which offers opportunities for the data analyst to do more with the IM than test pre-specified hypotheses \citep{cella.martin.probing}.  See \citet{martin.partial2, martin.partial3, martin.isipta2023, reimagined} for further details about IMs' properties, including connections to Bayesian/fiducial inference. 

In a Bayesian analysis, inference is based on summaries of the posterior distribution, e.g., probabilities of relevant hypotheses, expectations of loss/utility functions, etc.  All of these summaries boil down to integration involving the posterior density function.  Virtually the same statement holds for the possibilistic IM: the lower--upper probability pairs for relevant hypotheses, lower--upper expectations of loss/utility functions \citep{imdec}, etc., involve optimization of the IM's possibility contour $\pi_z$.  More precisely, the calculus underlying possibilistic reasoning is {\em Choquet integration} \citep[see, e.g.,][Appendix~C]{lower.previsions.book}.  The specific details (see Appendix~\ref{A:primer}) aren't directly relevant to the developments here, so it suffices to say Choquet integration relies heavily on optimization.  This observation has two important consequences: first, from the foundational point of view, compared to Bayes and other probabilistic frameworks, it's precisely the IM's optimization-driven Choquet integration that ensures validity is achieved; second, from a practical point of view, if reliable possibilistic inference requires optimization of $\pi_z$, then having efficient strategies for evaluating $\pi_z$ is imperative. 


While the IM construction is conceptually simple and its properties are quite strong, computation can be a challenge.  The key observation is that the sampling distribution of the relative likelihood $R(Z,\theta)$, under $\prob_\theta$, is rarely available in closed-form to facilitate exact computation of $\pi_{z}$.  So, instead, the go-to strategy is to approximate that sampling distribution using Monte Carlo at each value of $\theta$ on a sufficiently fine grid \citep[e.g.,][]{martin.partial2, hose.hanss.martin.belief2022}.  That is, the possibility contour is approximated as 
\begin{equation}
\label{eq:pi.naive}
\pi_{z}(\theta) \approx \frac1M \sum_{m=1}^M 1\{ R(Z_{m,\theta}, \theta) \leq R(z, \theta) \}, \quad \theta \in \TT,
\end{equation}
where $Z_{m,\theta}$ are independent copies of the data $Z$, specifically drawn from $\prob_\theta$, for $m=1,\ldots,M$.  The above computation is feasible at one or a few different $\theta$ values, but applications often require this to be carried out over a fine grid covering the relevant portion of the parameter space.  For example, identifying the confidence set in \eqref{eq:conf.set} requires finding which $\theta$'s satisfy $\pi_{z}(\theta) \geq \alpha$, and a naive approach is to compute the contour over a huge grid and then keep those that (approximately) meet the aforementioned condition.  This amounts to lots of wasted and expensive computations.  More generally, the relevant summaries of the IM output are Choquet integrals involving optimization, and doing so numerically requires many contour function evaluations.  This is a serious bottleneck, so new and more efficient computational strategies are desperately needed.

\section{Monte Carlo methods for IMs}
\label{S:mcim}

\subsection{Probabilistic approximation to a possibilistic IM}
\label{SS:insights}

Given a possibilistic IM $z \mapsto (\lPi_z,\uPi_z)$ with contour $\pi_z$, I'll refer to the level sets $C_\alpha(z) = \{ \theta \in \TT: \pi_z(\theta) \geq \alpha\}$ defined in \eqref{eq:conf.set} as $\alpha$-cuts. According to the theory reviewed in Section~\ref{S:background}, the $\alpha$-cut $C_\alpha(z)$ is a nominal $100(1-\alpha)$\% confidence region for $\Theta$.  Aside from the practical utility of confidence regions, there's much more information in the IM's $\alpha$-cuts. Towards this, the credal set $\cred(\uPi_z)$ corresponding to the IM's $\uPi_z$ is the set of (data-dependent) probability measures that it dominates, i.e., 
\[ \cred(\uPi_z) = \{ \prior_z \in \text{probs}(\TT): \text{$\prior_z(H) \leq \uPi_z(H)$ for all measurable $H$}\}, \]
where $\text{probs}(\TT)$ is the set of probability measures on $\TT$.  Then the claim is that the $\alpha$-cuts determine the IM's credal set.  Indeed, there's the following well-known characterization \citep[e.g.,][]{cuoso.etal.2001, destercke.dubois.2014} of the credal set $\cred(\uPi_z)$:
\begin{equation}
\label{eq:credal.1}
\prior_z \in \cred(\uPi_z) \iff \prior_z\{ C_\alpha(z) \} \geq 1-\alpha \quad \text{for all $\alpha \in [0,1]$}. 
\end{equation}
That is, a probability distribution $\prior_z$ belongs to the IM's credal set if and only if it's a confidence distribution, i.e., it assigns probability  $\geq 1-\alpha$ to each of the IM's $\alpha$-cuts.  More about this notion of ``confidence distributions'' and when Bayesian/fiducial distributions achieve equality in \eqref{eq:credal.1} can be found in \citet{martin.isipta2023, reimagined}.  At a soon-to-be-demonstrated practical level, elements of the IM's credal set have a relatively simple characterization, as the following theorem demonstrates.  The result is presented in a more general, not-IM-specific context because this simplifies the notation. 

\begin{thm}
\label{thm:credal.char}
Let $\uPi$ be a possibility measure on a space $\TT$, with corresponding contour $\pi$ and $\alpha$-cuts $C_\alpha = \{\theta \in \TT: \pi(\theta) \geq \alpha\}$, $\alpha \in (0,1)$.  Then $\prior \in \cred(\uPi)$ if and only if 
\begin{equation}
\label{eq:credal.2}
\prior(\cdot) = \int_0^1 \kernel^\beta(\cdot) \, \marg(d\beta), 
\end{equation}
for some Markov kernel $\kernel^\beta$ indexed by $\beta \in [0,1]$ such that $\kernel^\beta$ is fully supported on $C_\beta$, i.e., $\kernel^\beta(C_\beta) = 1$ for each $\beta \in [0,1]$ and some probability measure $\marg$ on $[0,1]$ such that a random variable with distribution $\marg$ is stochastically no smaller than $\unif(0,1)$. 
\end{thm}

\begin{proof}
See Appendix~\ref{A:proof}.
\end{proof}

I would be surprised if the characterization in Theorem~\ref{thm:credal.char} is genuinely new, but I've not found a reference for this result exactly.  A similar result is presented in \citet[][Theorem~2.1]{wasserman1990} characterizing the credal set of a belief function determined by a given distribution and set-valued mapping.  \citet[][Eq.~2.44]{hose2022thesis} also gives a special case of the right-hand side of \eqref{eq:credal.2} but makes no claims that his version is a complete characterization of $\cred(\uPi)$; what I specifically propose to do in \eqref{eq:Q.star} below closely matches the Hose's formula.  As an aside, choosing the kernel $\kernel^\beta$ to be a uniform distribution supported on $C_\beta$ is an idea that commonly appears in the literature, e.g., the {\em pignistic probability} \citep{smets.kennes.1994} or the {\em Shapley value} \citep{shapley1953} of a game.

The above theorem describes the contents of the credal set $\cred(\uPi)$ corresponding to a possibility measure $\uPi$.  A relevant follow-up question is if there exists an element in $\cred(\uPi)$ that's ``maximally consistent'' with $\uPi$, i.e., a distribution $\prior^\inn$ such that 
\begin{equation}
\label{eq:inner}
\text{equality in \eqref{eq:credal.1} holds: $\prior^\inn(C_\alpha) = 1-\alpha$ for each $\alpha \in [0,1]$}. 
\end{equation}
Such a probability distribution will be called an {\em inner probabilistic approximation} of $\uPi$.  The following corollary characterizes the inner probabilistic approximations.  

\begin{cor}
\label{cor:credal.char}
The equality in \eqref{eq:inner} is achieved if and only if  $\kernel^\beta(C_\alpha) = 0$ for all pairs $(\alpha,\beta)$ with $\beta < \alpha$---that is, if $\kernel^\beta$ is fully supported on $\partial C_\beta$---and $\marg=\unif(0,1)$.  
\end{cor}

\begin{proof}
See the end of the proof of Theorem~\ref{thm:credal.char}. 
\end{proof}


Note that $\prior^\inn$ in \eqref{eq:inner} may not exist, but this is rare in applications.  Indeed, if the contour $\pi$ doesn't span the entire range $[0,1]$, e.g., is bounded away from 0, then \eqref{eq:inner} can't be satisfied: some of the $\alpha$-cuts $C_\alpha$ would be equal to $\TT$ and every probability distribution supported on $\TT$ would assign probability $1 > 1-\alpha$ to these sets.  This manifests in the conditions of Corollary~\ref{cor:credal.char} because, in such a case, there's no ``boundary'' on which to support the kernel.  Even in these unfortunate cases, a probabilistic approximation can be made, but it would assign probability exceeding $1-\alpha$ to some of the $\alpha$-cuts $C_\alpha$.

\begin{illus}
As a simple-but-important application, which will prove to be useful in Section~\ref{SS:computation} below, consider the Gaussian possibility measure $\uPi$, with corresponding mean vector $m$ and covariance matrix $V$, whose contour is given by
\[ \pi(\theta) = 1 - F_d\{ (\theta-m)^\top V^{-1} (\theta-m) \}, \quad \theta \in \TT = \RR^d, \]
where $F_d$ is the $\chisq(d)$ distribution function.  The corresponding $\alpha$-cuts are ellipsoids 
\[ C_\alpha = \{ \theta \in \RR^d: (\theta-m)^\top V^{-1} (\theta-m) \leq F_d^{-1}(1-\alpha)\}, \quad \alpha \in [0,1]. \]
Of course, the normal distribution, $\nm_d(m, V)$, assigns probability $1-\alpha$ to each $\alpha$-cut, so that must be the inner probabilistic approximation; in fact, this is a tautology since the Gaussian possibility measure is defined as the outer possibilistic approximation of the Gaussian probability distribution \citep{imbvm.ext}.  For a direct connection to Corollary~\ref{cor:credal.char}, one first finds that the marginal distribution $\marg$ of $\pi(\Theta)$ is $\unif(0,1)$ under $\Theta \sim \nm_d(m,V)$.  Second, since a Gaussian random vector---or any elliptically symmetric random vector for that matter  \citep[e.g.,][Theorem~3.1]{hult.lindskog.2002}---can be written as $\Theta = m + R V^{1/2} U$, where $V^{1/2}$ is the Cholesky factor of $V$, $(R,U)$ are independent, with $U$ a uniform random vector on the unit sphere in $\RR^d$ and $R$ a non-negative random variable, the corresponding kernel $\kernel^\beta$ is the distribution of $m+\{F_d^{-1}(1-\beta)\}^{1/2} \, V^{1/2} \, U$, which is fully supported on the boundary $\partial C_\beta$ of $C_\beta$.  
\end{illus}

A direct implementation of the inner probabilistic approximation of the possibilistic IM based on the characterization in Corollary~\ref{cor:credal.char} is given in Appendix~\ref{A:direct}.  That direct approximation would generally be more expensive computationally compared to using a second-layer approximation as described in Sections~\ref{SS:outer}--\ref{SS:computation}.  No claims are made that this second-layer approximation is universal but, for the regular parametric models under consideration here, the results shown below based on this simpler approach are strikingly accurate, so the computationally cheaper solution is preferred.

\subsection{Approximating the probabilistic approximation}
\label{SS:outer}

To set the scene for a second-layer approximation, suppose that, in addition to the possibilistic IM output $\pi_z$ and $\uPi_z$, I have another data-dependent possibility contour $\pi_z^\out$ and corresponding possibility measure $\uPi_z^\out$.  These two sets of possibilistic output are related in the following critical way: $\pi_z(\theta) \leq \pi_z^\out(\theta)$ for all $\theta \in \TT$.  The construction of such a $\pi_z^\out$ will be discussed in Section~\ref{SS:computation}.  Ideally, $\pi_z$ and $\pi_z^\out$ would be close, subject to the dominance constraint, so I'll interpret $\pi_z^\out$ as an {\em outer approximation} of $\pi_z$.  The dominance property ``$\pi_z \leq \pi_z^\out$'' implies that $\uPi_z$ is {\em more specific} than $\uPi_z^\out$, i.e., $\uPi_z$ makes stronger claims than does $\uPi_z^\out$.  The interpretation of specificity might be easiest to understand through the relation induced between the credal sets: 
\[ \pi_z \leq \pi_z^\out \implies \cred(\uPi_z) \subseteq \cred(\uPi_z^\out). \]
In words, it takes fewer probability distributions to describe $\uPi_z$ than to describe $\uPi_z^\out$; consequently, there's less (epistemic) uncertainty in $\uPi_z$ than in $\uPi_z^\out$ and, hence, the former makes stronger claims than the latter.  It suffices to interpret $\Pi_z^\out$ as an IM that's (slightly) more conservative than $\uPi_z$.  Importantly, since the original IM with contour $\pi_z$ is valid in the sense of \eqref{eq:valid}, then so is the IM with contour $\pi_z^\out$.  For example, the $\alpha$-cuts $C_\alpha^\out(z) = \{\theta: \pi_z^\out(\theta) \geq \alpha\}$ of $\pi_z^\out$ are confidence sets, just like $C_\alpha(z)$. 
$\uPi_z^\out$ is appealing only if its computational benefits outweigh the loss of specificity/efficiency compared to $\uPi_z$.  This is often the case, as the numerical results below show.  

From here, my proposal is straightforward: approximate $\prior_z^\inn$ by $\prior_z^\star$, where the latter is defined as an inner probabilistic approximation of the outer possibilistic approximation $\uPi_z^\out$.  The relationship between $\prior_z^\inn$ and $\prior_z^\star$ can be described as follows:
\[ \prior_z^\inn \in \cred(\uPi_z) \subseteq \cred(\uPi_z^\out) \ni \prior_z^\star. \]
From this it's at least intuitively clear that, if $\pi_z^\out \approx \pi_z$ (subject to the dominance constraint), then $\prior_z^\star$ and $\prior_z^\inn$ should be similarly close.  Note, however, that $\prior_z^\star \not\in \cred(\uPi_z)$ in general.  That $\prior_z^\star$ might fail to be a member of $\cred(\uPi_z)$ is not an issue.  What matters is that $\prior_z^\star$ can maintain the reliability of $\uPi_z$ (or $\prior_z^\inn$) while being easier to compute.  In fact, since all the members of $\cred(\uPi_z)$ are, by definition, at least as tightly concentrated as $\prior_z^\inn$, an arguably safer strategy is to expand the set $\cred(\uPi_z)$ of candidate approximations a bit to create a better opportunity of finding one with the ``right'' concentration.  

Since $\prior_z^\star$ is itself an inner probabilistic approximation, a concrete expression for it, along with a strategy for numerically approximating it via Monte Carlo follows from the characterization result in Corollary~\ref{cor:credal.char}.  Indeed, $\prior_z^\star$ can be written as 
\begin{equation}
\label{eq:Q.star}
\prior_z^\star(\cdot) = \int_0^1 \kernel_z^\alpha(\cdot) \, d\alpha. 
\end{equation}
where $\kernel_z^\alpha$ is a distribution supported on $\partial C_\alpha^\out(z)$.  Whether $\prior_z^\star$ in \eqref{eq:Q.star} is an accurate approximation of and easier to compute than $\prior_z^\inn$ depends on the choice of $\pi_z^\out$.  Section~\ref{SS:computation} presents a simple construction of $\pi_z^\out$ that achieves both of these desiderata. 


\subsection{Implementing the second-layer approximation}
\label{SS:computation}

Here I construct the possibility contour $\pi_z^\out$ that dominates $\pi_z$ in the sense described in Section~\ref{SS:outer}.  Since the contour $\pi_z^\out$ is determined by its $\alpha$-cuts $\{C_\alpha^\out(z): \alpha \in [0,1]\}$, the presentation that follows focuses on the latter.  

For the relative likelihood-based possibilistic IM defined in Section~\ref{S:background}, under the usual regularity conditions, \citet{imbvm.ext} established an asymptotic Gaussianity result that says, if the sample size is large, then 
\begin{equation}
\label{eq:approx.gaussian}
\pi_z(\theta) \approx 1 - F_d\bigl( (\theta - \hat\theta_z)^\top J_z (\theta - \hat\theta_z) \bigr), \quad \theta \in \TT \subseteq \RR^d, 
\end{equation}
where $\hat\theta_z$ is the maximum likelihood estimator, $J_z$ is the  $d \times d$ observed Fisher information matrix, and, again, $F_d$ is the $\chisq(d)$ distribution function.  Consequently, the $\alpha$-cuts $C_\alpha(z)$ of $\pi_z$ are approximately ellipsoidal.  This suggests a basic elliptical form for $C_\alpha^\out(z)$ and hence for $\pi_z^\out$, which is asymptotically correct/optimal, and is shown to be quite accurate even for small samples.  As I'll explain below, assuming that $C_\alpha(z)$ can be roughly approximated by an ellipsoid is not equivalent to assuming Gaussianity. 

Note that the approximation \eqref{eq:approx.gaussian} also holds for different parametrizations. In fact, it might be more accurate in one parametrization than in another, e.g., non-negativity constraints on $\theta$ suggest applying a Gaussian approximation on the $\log\theta$-scale.  What's described below on the $\theta$-scale can also be applied following a reparametrization, and I freely reparametrize as appropriate in my numerical examples below.  

For constructing $C_\alpha^\out(z)$, I leverage an idea advanced by \citet{imvar.ext}.  Below I describe their idea---different from how they described it---and the relevant technical details are given in Appendix~\ref{A:imvar}.  Define an eigenvalue-scaled version $J_z(\sigma)$ of the observed information $J_z$ as follows: if $J_z = E \Lambda E^\top$ is the spectral decomposition, then 
\[ J_z(\sigma) = E \, \text{diag}(\sigma^{-1}) \, \Lambda \, \text{diag}(\sigma^{-1}) \, E^\top, \]
where $\sigma^{-1}$ is the entry-wise reciprocal of the vector $\sigma \in (0,\infty)^d$, and $\text{diag}(\cdot)$ is the operator that takes a vector argument to a diagonal matrix with that vector on the diagonal.  Then $\sigma$ acts like a vector of scaling parameters that can stretch or contract the level sets of the quadratic form $\theta \mapsto (\theta - \hat\theta_z)^\top J_z(\sigma)^{-1} (\theta - \hat\theta_z)$ in different directions.  Next, define the $\sigma$-specific $\alpha$-cut for a Gaussian $\nm_d( \hat\theta_z, J_z(\sigma)^{-1} )$ possibility contour 
\begin{align}
C_\alpha^\sigma(z) & = \{\theta: 1 - F_d\bigl( (\theta - \hat\theta_z)^\top J_z(\sigma) (\theta - \hat\theta_z) \bigr) \geq \alpha \} \notag \\
& = \{\theta: (\theta-\hat\theta_z)^\top J_z(\sigma) (\theta - \hat\theta_z) \leq F_d^{-1}(1-\alpha)\}, \label{eq:var.cut}
\end{align}
Since $C_\alpha(z)$ is, for each $\alpha$ and $z$, approximately elliptical with center $\hat\theta_z$ and shape roughly that of a $J_z$-quadratic form, one expects that there exists a scaling vector $\sigma=\sigma(z,\alpha)$ such that, at least approximately,  
\[ C_\alpha^{\sigma(z,\alpha)}(z) \supseteq C_\alpha(z), \quad \alpha \in [0,1]. \]
A cartoon illustration of this bounding is shown in Figure~\ref{fig:bean}, and the algorithm proposed in \citet{imvar.ext} for finding such a $\sigma(z,\alpha)$ is given in Appendix~\ref{A:imvar}.  There are two important points worth mentioning here.  First, the algorithm developed in \citet{imvar.ext} is efficient in the sense that evaluating $\sigma(z,\alpha)$ has complexity is $O(d)$ compared to $O(e^d)$ for the naive approach in \eqref{eq:pi.naive}.  Second, note that it's not a single scaling factor $\sigma$ that achieves the containment in the above display for each $\alpha$ as would be the case if $\pi_z$ were exactly Gaussian; instead, it's an $\alpha$-specific adjustment which, at least intuitively, suggests more flexible ``scale mixture of normals'' form.  

\begin{figure}[t]
\begin{center}
\scalebox{0.9}{
\begin{tikzpicture}
\filldraw[color=black!20, fill=black!5, rotate=140] (-1.25, 1.25) to[closed, curve through = { (1, 1.25) (2.5, 2.25) (2.25, -0.65) (-3.3, -0.15) (-2.5, 0.65) }] (-1.25, 1.25); 
\filldraw[color=black, fill=black] (0,0) circle (1.5pt) node[anchor=east]{$\hat\theta_z$};
\draw (1.1, -0.65) circle (0pt) node[anchor=west]{$C_\alpha(z)$}; 
\draw[rotate=162] (0,0) ellipse (4.15 and 1.86);
\draw[dashed, rotate=162] (0,0) ellipse (4.75 and 1.15);
\draw (2.40, 1.1) circle (0pt) node[anchor=west]{$\partial C_\alpha^{\sigma(z,\alpha)}(z)$}; 
\end{tikzpicture}
}
\end{center}
\caption{Approximating the IM contour's $\alpha$-cut $C_\alpha(z)$ by an ellipse $C_\alpha^\sigma(z)$ with a ``good'' choice of $\sigma$ (solid) and with a ``bad'' choice of $\sigma$ (dashed).}
\label{fig:bean}
\end{figure}

Putting everything together, my specific proposal is to define the $C_\alpha^\out(z) := C_\alpha^{\sigma(z,\alpha)}(z)$, which, in turn, leads to the probabilistic approximation $\prior_z^\star$ as in \eqref{eq:Q.star}.  And the mixture form \eqref{eq:Q.star} immediately suggests a Monte Carlo method for evaluating various $\prior_z^\star$-probabilities; see Section~\ref{SS:back.to.poss}.  The whole procedure is succinctly described in Algorithm~\ref{algo:inner}.  A brief asymptotic analysis justifying the accuracy of the proposed approximation is presented in  Remark~\ref{re:asymptotics} of Appendix~\ref{A:remarks}.  A few conceptual points are given next. 
\begin{itemize}
\item Of course, the two sampling steps in Algorithm~\ref{algo:inner} are trivial.  The only non-trivial step, namely ``evaluate $\sigma(z,\alpha)$,'' is easy to miss.  This involves an $\alpha$-specific run of the algorithm summarized in Appendix~\ref{A:imvar} which, fortunately, typically only involves a few updates and a few Monte Carlo evaluations of the IM contour $\pi_z$ per update.  So, this is easy to compute and has complexity scaling linearly in dimension $d$.  Moreover, it's easy to parallelize the evaluation of $\sigma(z,\alpha)$ over multiple $\alpha$'s, dramatically reducing computation time; see, also, Remark~\ref{re:parallel} in Appendix~\ref{A:remarks}. 
\item The benefit to using the elliptical approximations should now be apparent, since sampling from $\prior_z^\star$ is effectively no different than sampling from a Gaussian as in the illustration towards the end of Section~\ref{SS:insights}.  But remember that $\prior_z^\star$ is generally not Gaussian---the $\alpha$-adaptivity gives it more flexibility.  
\item Despite the aforementioned adaptivity, forcing $\partial C_\alpha^\out(z)$ to be elliptically-shaped potentially implies some inefficiency in small samples, as Figure~\ref{fig:bean} reveals.  Theoretically and empirically, this restriction is rather mild, but there may be applications in which an elliptical shape might not be appropriate.  Fortunately, the elliptical structure isn't necessary to deploy the solution proposed here; see Appendix~\ref{A:direct}.  
\end{itemize} 


\begin{algorithm}[t]
\SetAlgoLined
Independent draws from $\prior_z^\star$ in \eqref{eq:Q.star} are obtained by repeating the following steps:
\begin{enumerate}
\item Sample $A \sim \unif(0,1)$;
\vspace{-2mm}
\item Given $A=\alpha$, evaluate $\sigma(z,\alpha)$, and then sample $\Theta$ just as described in the Gaussian example at the end of Section~\ref{SS:insights}, i.e., 
\[ (\Theta \mid A=\alpha) \overset{\text{\tiny dist}}{\, = \,} \hat\theta_z + \{ F_d^{-1}(1-\alpha) \}^{1/2} \, J_z^{-1}\bigl( \sigma(z,\alpha) \bigr)^{1/2} \, U, \]
where half-power of a matrix means the Cholesky factor and $U$ is uniformly distributed on the sphere in $\RR^d$. 
\end{enumerate}
\caption{Sampling from the probabilistic approximation $\prior_z^\star$ of $\uPi_z$.}
\label{algo:inner}
\end{algorithm}

Two final remarks to put the contribution into perspective will close out this subsection.  First, Cella and Martin's goal was the same as that in the present paper.  What they offered as their final, albeit incomplete solution was a simple, $\alpha$-specific Gaussian probabilistic approximation, namely, $\prior_z^{\text{\sc cm},\alpha} = \nm_d(\hat\theta_z, J_z(\sigma(z,\alpha))^{-1})$, which is then transformed into a possibility measure as described in Section~\ref{SS:back.to.poss} below.  Their approximation is only designed to capture a certain aspect of the original IM, in particular, its $\alpha$-cut for the one specified $\alpha$, so it need not be (although often is) a good approximation of $\pi_z$ overall.  {\em The goal is obviously a good overall approximation to $\pi_z$} and \citet[][Sec.~6]{imvar.ext} acknowledge this limitation: ``finding a way to stitch together these $\alpha$-specific IM approximations is an important open problem.''  The present paper's contribution is filling this significant gap in the IM computational developments.  

Second, my proposed solution has certain aspects in common with the {\em calibrated bootstrap} developed in \citet{calibrated.boostrap}.  In fact, my solution was actually inspired by the calibrated bootstrap.  Specialized points are needed to make this comparison so, for the sake of brevity here, I defer these details to Remark~\ref{re:caliboot} in Appendix~\ref{A:remarks}.


\subsection{Back to a possibilistic IM} 
\label{SS:back.to.poss}

The previous subsections described a probabilistic approximation $\prior_z^\star$ of the possibilistic IM, which could be sampled with relative ease.  Relevant summaries could be readily approximated using the samples from $\prior_z^\star$, and subsequently used to approximate the IM's summaries.  The challenge is that, for each summary, there are mathematical limits to how well a probabilistic approximation can match the target possibilistic one; for some summaries, the approximation is accurate while for others it's not.  The question is if one can convert the probabilistic approximation $\prior_z^\star$ to a possibilistic approximation that better represents the target IM while retaining the computational efficiency.  

One simple way forward is to note that the probabilistic approximation just described can be interpreted as a possibility-to-probability transform, which has an inverse probability-to-possibility transform.  Let $r_z: \TT \to \RR$ be a (possibly data-dependent) ranking function on the parameter space, where larger values of $r_z(\theta)$ indicate that $\theta$ is ``higher ranked'' in some meaningful sense.  Good examples of ranking functions are, first, the density function $q_z^\star$ corresponding to the distribution $\prior_z^\star$ and, second, the likelihood function $L_z$ of the underlying model.  With both of these choices, the meaningfulness of the ranking is clear.  Taking the ranking function to be the density $q_z^\star$ of $\prior_z^\star$ gives the tightest contours, similar to how highest density sets have smallest volume among those with a fixed probability content, but ``tightest contours'' isn't the objective.  Moreover, there may be computational challenges associated with use of the density-based ranking.  After all, $q_z^\star$ would need to be estimated using, say, kernel methods applied to the samples from $\prior_z^\star$ as described above, which is non-trivial when $d > 1$.  

In any case, for a given ranking function $r_z$, the probability-to-possibility transform of $\prior_z^\star$ returns the possibility contour 
\begin{equation}
\label{eq:p2p}
\hat\pi_z(\theta) = \prior_z^\star\bigl\{ r_z(\Theta) \leq r_z(\theta) \bigr\}, \quad \theta \in \TT. 
\end{equation}
To see that $\sup_\theta \hat\pi_z(\theta) = 1$ and, hence, that this transformation defines a genuine possibility contour, just take $\theta$ equal (or converging) to a ``highest ranked'' value according to $r_z$.  
In the examples that follow, the approximation in \eqref{eq:p2p} is called a {\em stitched IM contour}.

Computationally, the contour $\hat\pi_z$ in \eqref{eq:p2p} is far more efficient than that in \eqref{eq:pi.naive}.  The reason being that the Monte Carlo samples in the former are fixed whereas, in the latter, the samples generally depend on the $\theta$ at which the contour is being evaluated.  But is $\hat\pi_z$ is a good approximation of $\pi_z$?  Towards this, consider the genuine inner probabilistic approximation $\prior_z^\inn$ of $\uPi_z$ as described in Section~\ref{SS:insights}.  Then it's easy to see that 
\[ \prior_z^\inn\{ \pi_z(\Theta) \leq \pi_z(\theta) \} = \pi_z(\theta), \quad \theta \in \TT, \]
where the left-hand side above is the probability-to-possibility transform of $\prior_z^\inn$ with ranking function $r_z = \pi_z$.  So, with a proper choice of ranking function, it's theoretically possible to recover the original IM from this.  The downside is that, of course, if $\pi_z$ was available to serve as a ranking function, then there'd be no need for any of these approximations.  The same argument can be made with $\prior_z^\inn$ and $\pi_z$ in the above display replaced by $\prior_z^\star$ and $\pi_z^\out$, respectively, but the again the issue is that $\pi_z^\out$ isn't readily available.  So, it's unrealistic to expect that a stitched IM can exactly recover the original IM, but large-sample theory and my own experience suggest that using the likelihood-based ranking $r_z=L_z$ in \eqref{eq:p2p} is a good standard choice.


\subsection{Handling nuisance parameters}
\label{SS:nuisance}

It's often the case in applications that the full model parameter $\Theta$ is not the quantity of interest.  That is, one is primarily interested in some feature $\Phi = f(\Theta)$ of the full model parameter; for example, in analysis of variance, we often care mainly about the treatment effect parameters and not so much about the error variance.  \citet{martin.partial3} presents two different strategies for marginalization in the possibilistic IM framework: the {\em extension-based} strategy is computationally and conceptually simple but generally inefficient whereas the {\em profile-based} strategy involves some problem-specific tailoring but is generally more efficient.  Appendix~\ref{A:marginal} provides some details about these two marginalization strategies.  Importantly, both strategies guarantee validity, so the more efficient profiling strategy is recommended.  But use of the proposed inner probabilistic approximations adds a new twist to the marginalization story which I'll discuss below.  

The probabilistic approximation described above suggests an obvious marginalization strategy, namely, mapping the sample $\{\Theta_m: m=1,\ldots,M\}$ from $\prior_z^\star$ to a corresponding sample $\{\Phi_m: m=1,\ldots,M\}$ from the marginal distribution $\prior_z^{\star f}$, where $\Phi_m = f(\Theta_m)$.  From there, one could apply a probability-to-possibility transformation like in Section~\ref{SS:back.to.poss}.  This will be referred to as an {\em indirect approach}, in contrast to the direct approach described below.  While this indirect approximation is natural and simple to carry out, a relevant question is if it's approximating a valid marginal IM.  The answer to this question depends on the mapping $f$.  \citet[][Theorem~6]{reimagined} shows that if $f$ is boundary-preserving---i.e., the boundary of $f\{ C_\alpha(z) \}$ is the image under $f$ of the boundary of $C_\alpha(z)$ for each $z$ and $\alpha$, like for linear functions $f$---then $\prior_z^{\star f}$ is an inner probabilistic approximation of the extension-based marginal possibilistic IM for $\Phi$.  In such cases, the indirect approach described at the start of this paragraph provides an accurate approximation of the valid extension-based marginal IM for $\Phi$.  

The indirect approach need not be a bad solution when the mapping isn't boundary-preserving since, e.g., a non-boundary-preserving map $f$ might be approximately linear in the region where $\prior_z^\star$ is concentrated.  In general, however, one can't guarantee that this leads to an approximation of a valid marginal IM.  The point is that probabilistic marginalization through a non-linear function tends to be too concentrated and the resulting inferences are over-confident, i.e., confidence regions are too narrow, and no probability-to-possibility transform can correct this. 

A safer and generally more efficient marginalization strategy is available, but requires special tailoring to the specific feature $\Phi=f(\Theta)$ of interest.  First, construct the marginal IM possibilistic IM using either the extension- or profiling-based strategies described in Appendix~\ref{A:marginal}; of course, since profiling is more efficient, that would be my recommended strategy.  Second, construct the corresponding inner probabilistic approximation to the marginal possibilistic IM.  This is what I refer to as the {\em direct approach}---``direct'' in the sense that the approximation is applied to the marginal IM for the specific quantity of interest.  The point is that the indirect approach, which relies on probabilistic marginalization, is easy but not guaranteed to be reliable; to ensure reliability, the direct approach instead carries out possibilistic marginalization, which is provably reliable, and then applies the probabilistic approximation to the marginal IM.

\section{Low-dimensional illustrations}
\label{S:examples}

I'll start with some relatively simple low-dimensional examples where the proposed approximation can be visualized and, at least in some cases, the naive computation in \eqref{eq:pi.naive} can be carried out and the proposed approximation's accuracy can be visually assessed.  A higher-dimensional example is presented in Section~\ref{S:application} below.  

\begin{ex}
\label{ex:gamma}
Let $Z=(Z_1,\ldots,Z_n)$ denote an iid sample from a gamma distribution with unknown parameter $\Theta=(\Theta_1,\Theta_2)$, where $\Theta_1 > 0$ and $\Theta_2 > 0$ are the unknown shape and scale parameters, respectively.  It is straightforward to get the naive approximation \eqref{eq:pi.naive} at any particular value of the parameter, but prohibitively expensive to carry this out over a sufficiently fine grid that spans the plausible pairs $(\theta_1,\theta_2)$.  So, there's a need for more computationally efficient approximation methods, and the Monte Carlo strategy of Section~\ref{S:mcim} above fits the bill.  For this illustration, I'll work with the data presented in Example~3 of \citet{fraser.reid.wong.1997}, which consists of the survival time (in weeks) for $n=20$ rats exposed to a certain amount of radiation, modeled by a gamma distribution with unknown shape and scale parameters.  Figure~\ref{fig:gamma.hist} summarizes the 5000 samples of $(\log\Theta_1, \log\Theta_2)$ from the proposed probabilistic approximation $\prior_z^\star$ developed in Section~\ref{S:mcim}.  This reveals that the distribution is at least approximately Gaussian.  Figure~\ref{fig:gamma.pl}(a) shows the exact joint contour function for $(\Theta_1,\Theta_2)$ along with two stitched IM approximations: one takes the ranking function to be a bivariate Gaussian density function and the other takes it to be the gamma likelihood function.  Notice that the likelihood-ranking-based approximation almost perfectly matches the exact contour.  For comparison, the exact contour requires more than 10 minutes of computation time, while the new Monte Carlo approximation is completed in just a few seconds.  

A practically relevant and surprisingly challenging follow-up question concerns inference on the mean $\Phi = \Theta_1\Theta_2$ of the gamma distribution.  IM solutions that offer exactly valid inference on $\Phi$ are presented in \citet{immarg} and \citet{martin.partial3}, but the relevant computations can be a burden.  Section~\ref{SS:nuisance} describes two approaches---direct and indirect---for carrying out this marginalization, and Figure~\ref{fig:gamma.pl}(b) shows the two corresponding contours for $\Phi$, along with the exact profile-based marginal IM contour.   Clearly, both the direct and indirect approaches are quite accurate in this case and, indeed, the direct approach is more accurate overall, with the indirect approach's contour being a bit too narrow compared to the exact contour.  This narrowness is inevitable because the marginal distribution for the mean $\Phi$ derived from $\prior_z^\star$ is not an inner probabilistic approximation of a valid marginal IM.  But what the indirect approach might lack in accuracy, it makes up for in simplicity. 
\end{ex}

\begin{figure}[t]
\begin{center}
\subfigure[Samples from $\prior_z^\star$]{\scalebox{0.4}{\includegraphics{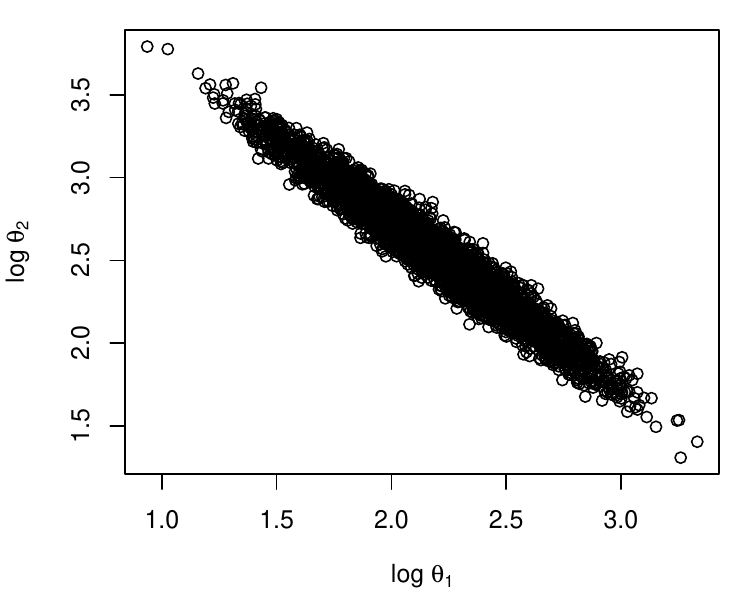}}}
\subfigure[Histogram of $\log\Theta_1$]{\scalebox{0.4}{\includegraphics{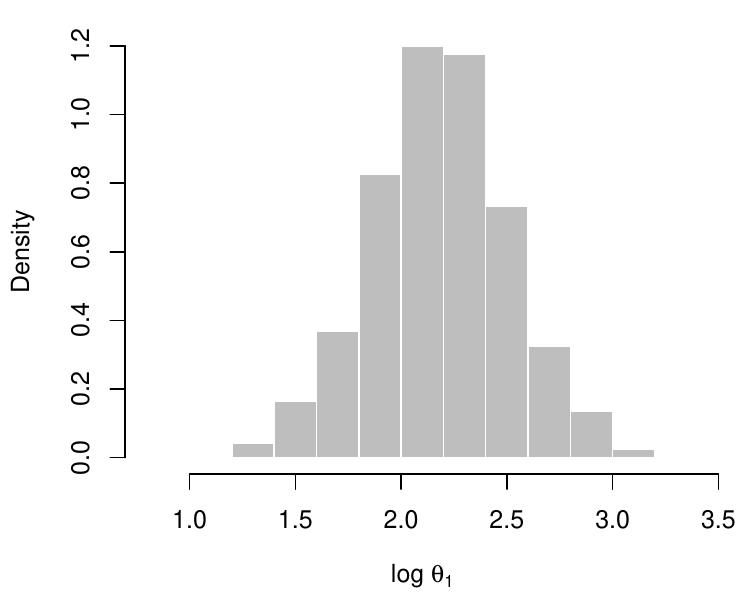}}}
\subfigure[Histogram of $\log\Theta_2$]{\scalebox{0.4}{\includegraphics{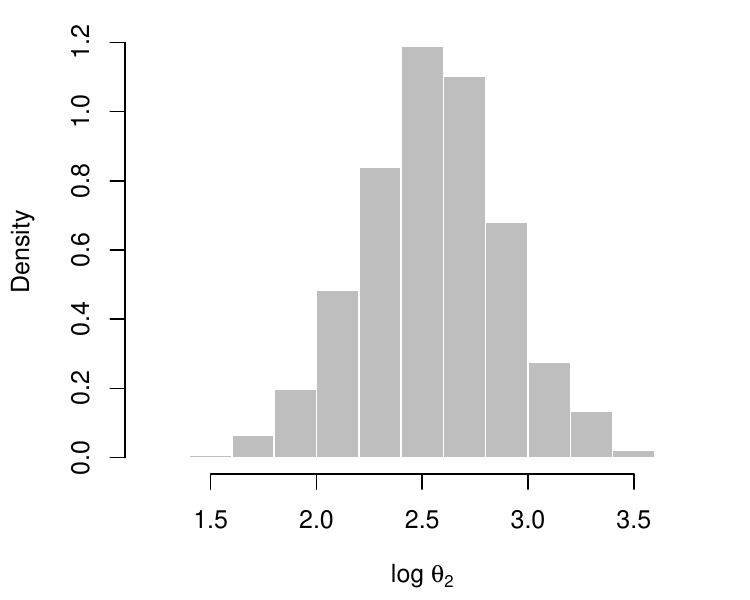}}}
\end{center}
\caption{Summary of the samples $(\log\Theta_1,\log\Theta_2)$ from $\prior_z^\star$ in Example~\ref{ex:gamma}.}
\label{fig:gamma.hist}
\end{figure} 

\begin{figure}[t]
\begin{center}
\subfigure[Joint contours for $(\Theta_1,\Theta_2)$]{\scalebox{0.55}{\includegraphics{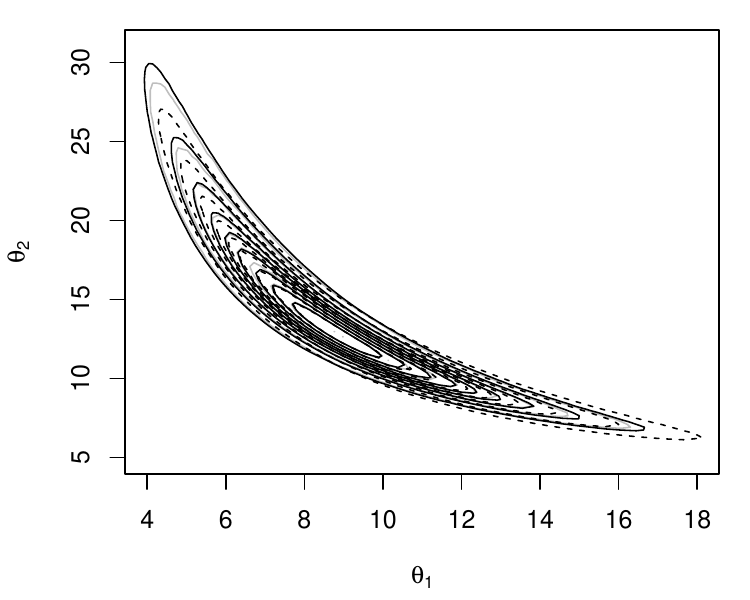}}}
\subfigure[Marginal contours for the mean $\Theta_1\Theta_2$]{\scalebox{0.55}{\includegraphics{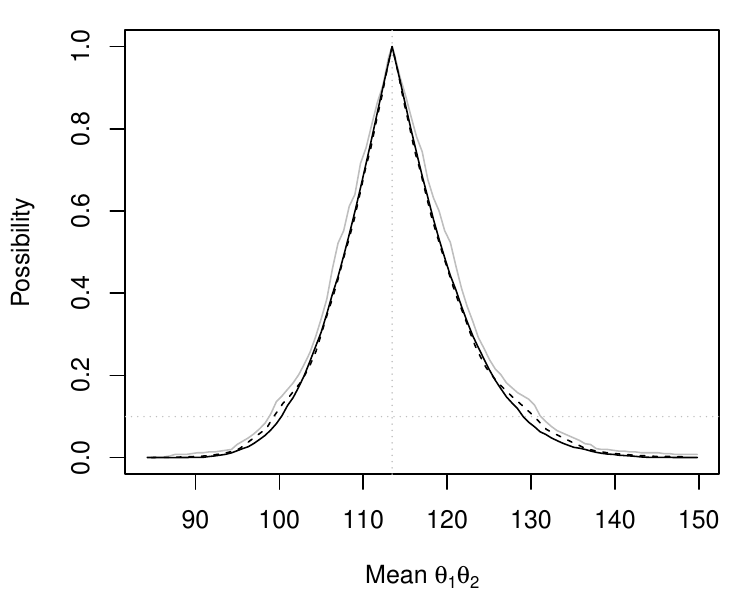}}}
\end{center}
\caption{Results for gamma model fit in Example~\ref{ex:gamma}.  Panel~(a) shows the exact IM contour (gray) for the gamma shape and scale parameters $(\Theta_1,\Theta_2)$ and two approximations: Gaussian density-based ranking (dashed) and likelihood-based ranking (solid).  Panel~(b) shows the exact marginal IM contour (gray) for the gamma mean and two approximations as described in the text: direct (dashed) and indirect (solid).}
\label{fig:gamma.pl}
\end{figure}

\begin{ex}
\label{ex:logistic}
The data presented in Table~8.4 of \citet{ghosh-etal-book} concerns the relationship between exposure to chloracetic acid and mouse mortality. A simple logistic regression model can be fit to relate the binary death indicator ($y$) with the levels of exposure ($x$) to chloracetic acid for the dataset's $n=120$ mice.  That is, data $Z$ consists of pairs $Z_i = (X_i, Y_i)$, for $i=1,\ldots,n$, and a conditionally Bernoulli model for $Y_i$, given $X_i$, with mass function 
\[ p_\theta(y \mid x) = F(\theta_1 + \theta_2 x)^y \, \{1 - F(\theta_1 + \theta_2 x)\}^{1-y}, \quad \theta=(\theta_1, \theta_2) \in \RR^2, \]
where $F(u) = (1 + e^{-u})^{-1}$ is the logistic distribution function.  The corresponding likelihood cannot be maximized in closed-form, but this is easy to do numerically, and the maximum likelihood estimator and the corresponding observed information matrix lead to the asymptotically valid inference reported by standard statistical software.  
For exact inference, however, the computational burden is heavier: evaluating the exact IM contour over a sufficiently fine grid of $\theta$ values is prohibitively expensive.  As an alternative, the Monte Carlo sampling method presented in Section~\ref{S:mcim} is easy to implement and runs in a matter of seconds.  Figure~\ref{fig:logistic}(a) shows the 5000 Monte Carlo samples of $(\Theta_1,\Theta_2)$ from $\prior_z^\star$ along with the stitched IM contour based on a Gaussian density ranking, closely agreeing with the approximations in \citet{imvar.ext} and \citet{imbvm.ext}.  

A specific and practically relevant question concerns the chloracetic acid exposure level required to make the death probability 0.5.  This amounts to setting $F(\theta_1 + \theta_2 x) = 0.5$ and solving for $x$; the solution is $\lambda = -\theta_1/\theta_2$, and the true value $\Lambda = -\Theta_1/\Theta_2$ is called the median lethal dose, or LD50.  For inference on $\Lambda$, I use the previously-obtained samples $(\Theta_1,\Theta_2)$ from $\prior_z^\star$ and use a Gaussian density ranking to construct and approximate possibility contour for $\Lambda$.  This curve is shown in Figure~\ref{fig:logistic}(b) and, again, it closely agrees with the asymptotic approximations in \citet{imbvm.ext}. 
\end{ex}


\begin{figure}[t]
\begin{center}
\subfigure[Joint contour for $(\Theta_1,\Theta_2)$]{\scalebox{0.55}{\includegraphics{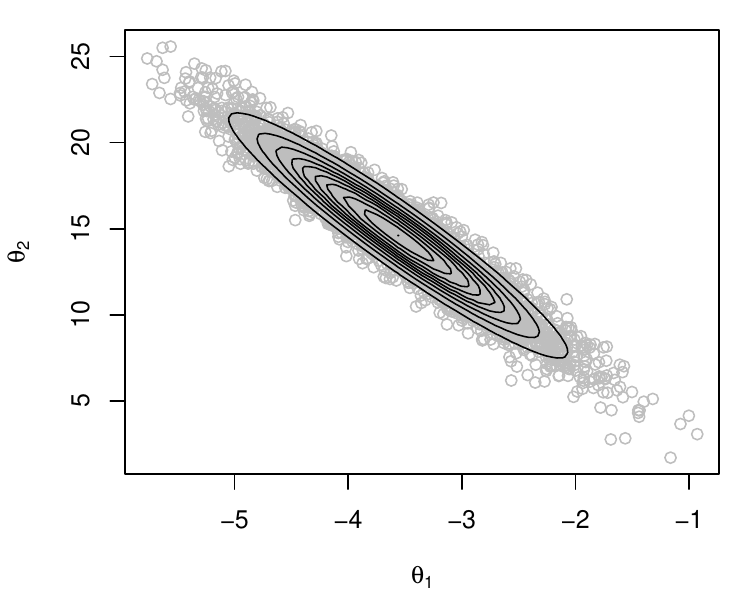}}}
\subfigure[Marginal contour for $\Lambda$]{\scalebox{0.55}{\includegraphics{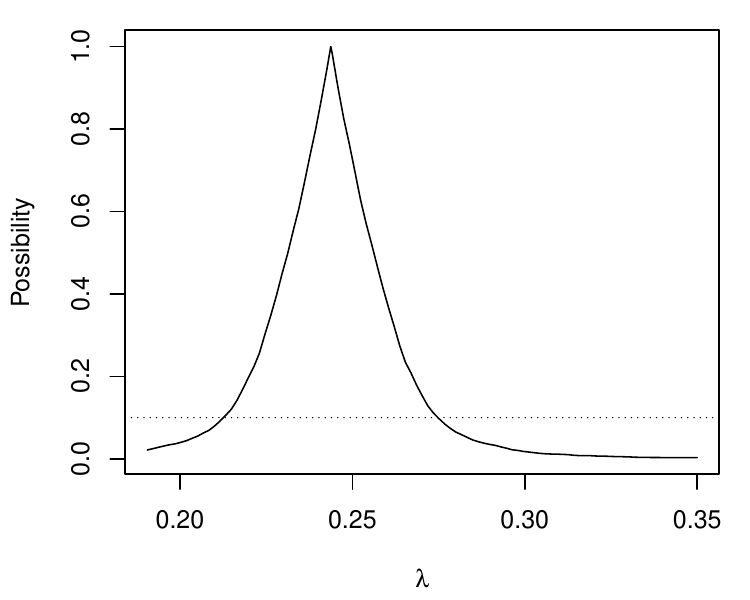}}}
\end{center}
\caption{Summary of the results in Example~\ref{ex:logistic}. Panel~(a) shows the samples of $(\Theta_1,\Theta_2)$ from $\prior_z^\star$ and the corresponding approximate possibility contour. Panel~(b) shows the approximate possibility contour for the median lethal dose $\Lambda$.}
\label{fig:logistic}
\end{figure} 

\begin{ex}
\label{ex:censored}
The examples above are non-trivial when it comes to exact (marginal) inference, but they still fall under the umbrella of ``standard'' models.  As model complexity increases, obviously so too does the computational burden of exact inference.  One example of a relatively complex model is that required for modeling censored data: on the one hand, data from the population under investigation are corrupted, which adds complexity, and, on the other hand, little is known about the corruption process so this should be modeled nonparametrically, which also adds complexity.  More specifically, I consider time-to-death data on ovarian cancer patients from a clinical trial that took place from 1974 to 1977 \citep{edmonson.etal.1979}; this data is contained in the {\tt ovarian} data set in the R package {\tt survival} \citep{survival-package}. Among other things, this data set includes survival times for $n=26$ patients that entered the study with Stage II or IIIA cancer and were treated with either cyclophosphamide alone or cyclophosphamide with adriamycin. Of these patients, 14 survived to the end of the study---so their survival times were right censored---and 12 unfortunately died.  In such cases, it's common to have a parametric model for the survival times, which is adjusted in a nonparametric way to accommodate censoring.  Let $Y_i$ denote the actual survival time of patient $i$, which may or may not be observed.  Assign these survival times a statistical model $\{\prob_\theta: \theta \in \TT\}$, where the true-but-unknown value $\Theta$ of the model parameter is the target.  Let $C_i$ denote the censoring time for patient $i$, which is a random variable.  Then, under the assumption of random right censoring, the observed data $Z$ consists of $n$ iid pairs $Z_i = (X_i, T_i)$, where 
\begin{equation}
\label{eq:censoring.rule}
X_i = \min( Y_i, C_i ) \quad \text{and} \quad T_i = 1(Y_i \leq C_i), \quad i=1,\ldots,n,
\end{equation}
where $T_i=1$ if the observation is an event time and $T_i=0$ if it's a censoring time.  The goal is to infer the unknown $\Theta$, but with the censoring-corrupted survival times.  The likelihood function for the observed data is 
\[ L_{z}(\theta, G) = \prod_{i=1}^n g(x_i)^{1-t_i} \, \{1 - G(x_i)\}^{t_i} \times \prod_{i=1}^n p_\theta(x_i)^{t_i} \{ 1 - P_\theta(x_i) \}^{1-t_i}, \]
which depends on both the generic value $\theta$ of the true unknown model parameter $\Theta$ for the concentrations and on the generic value $G$ of the true unknown censoring level distribution $\mathsf{G}$.  In the above expression, $g$ and $p_\theta$ are density functions for the censoring and concentration distributions, and $G$ and $P_\theta$ are the corresponding distribution functions.  

Following the review in Section~\ref{S:background} and details in Appendix~\ref{A:marginal}, given that $\mathsf{G}$ is a nuisance parameter, the natural strategy is to work with a relative profile likelihood,  
\[ R^\text{\sc pr}(z, \theta) = \frac{\prod_{i=1}^n p_\theta(x_i)^{t_i} \{ 1 - P_\theta(x_i) \}^{1-t_i}}{\prod_{i=1}^n p_{\hat\theta_z}(x_i)^{t_i} \{ 1 - P_{\hat\theta}(x_i) \}^{1-t_i}}, \quad \theta \in \TT, \]
where $\hat\theta_z$ is the maximum likelihood estimator of $\Theta$.  The {\em distribution} of the relative profile likelihood still depends on the nuisance parameter $\mathsf{G}$, so when we define the possibilistic IM contour by validifying the relative profile likelihood, we get 
\[ \pi_z(\theta) = \sup_G \prob_{\theta, G} \bigl\{ R^\text{\sc pr}(Z, \theta) \leq R^\text{\sc pr}(z, \theta) \bigr\}, \quad \theta \in \TT. \]
\citet{imcens} proposed a novel strategy wherein a variation on the Kaplan--Meier estimator \citep[e.g.,][]{kaplanmeier, km.book} is used to obtain a $\widehat G$, and then the contour above is approximated by 
\begin{equation}
\label{eq:censored.boot.pl}
\check\pi_z(\theta) = \prob_{\theta, \widehat G} \bigl\{ R^\text{\sc pr}(Z, \theta) \leq R^\text{\sc pr}(z, \theta) \bigr\}, \quad \theta \in \TT. 
\end{equation}
Evaluation of the right-hand side via Monte Carlo boils down to sampling censoring levels from $\widehat G$, sampling concentration levels from $\prob_\theta$, and then constructing new data sets according to \eqref{eq:censoring.rule}.  While this procedure is conceptually simple, naive implementation over a sufficiently fine grid of $\theta$ values is very expensive.  Fortunately, the proposed Monte Carlo strategy can be readily applied to sample from a probabilistic approximation of $\check\pi_z$, from which approximately valid inference on $\Theta$ can be obtained.  While only first-order asymptotic validity of this IM could be established in \citet{imcens}, the empirical results presented there are strikingly accurate.

As is common in the time-to-event data analysis literature, I'll take a Weibull model for the survival times, where the density function is 
\[ p_\theta(y) = (\theta_1/\theta_2) \, (y / \theta_2)^{\theta_1 - 1} \, \exp\{ -(y / \theta_2)^{\theta_1} \}, \quad y > 0, \]
where $\theta=(\theta_1,\theta_2)$ is the unknown parameter, with $\theta_1 > 0$ and $\theta_2 > 0$ the shape and scale parameters, respectively.  Figure~\ref{fig:weibull.pl} shows the samples of $(\Theta_1,\Theta_2)$ from the inner probabilistic approximation $\prior_z^\star$ of the IM defined by \eqref{eq:censored.boot.pl}.  This plot also shows two approximate contours for $\Theta$: one is based solely on the asymptotic Gaussianity of the possibilistic IM as demonstrated in \citet{imbvm.ext} and the other is based on the samples from $\prior_z^\star$ and the Gaussian density ranking.  The former asymptotic contour assumes ``$n \approx \infty$'' which is difficult to justify with $n=26$ observations, so the latter contour, which is more diffuse in this application, offers more trustworthy inference.  It's also very similar to the results presented in \citet{imcens}, but with only a tiny fraction of the computational cost.  

\begin{figure}[t]
\begin{center}
\scalebox{0.6}{\includegraphics{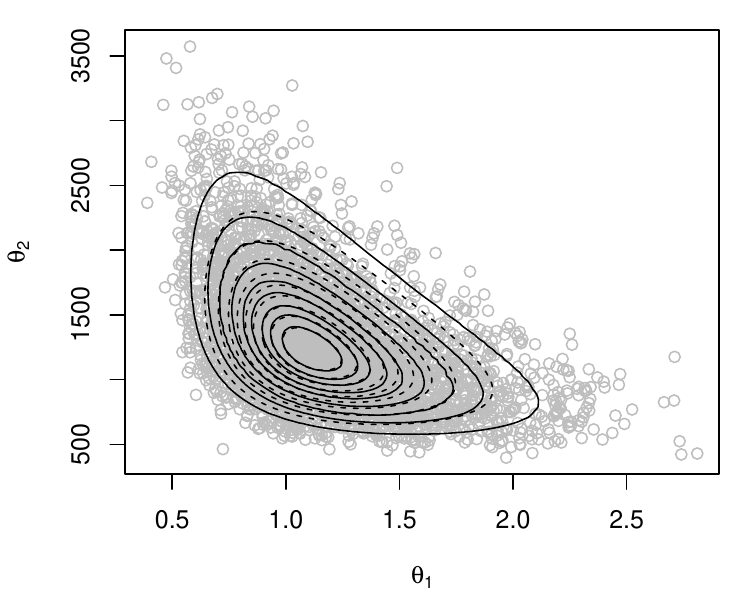}}
\end{center}
\caption{Plot of samples $(\Theta_1,\Theta_2)$ of the Weibull shape and scale parameters from the inner probabilistic approximation $\prior_z^\star$ along with two corresponding joint contours: Gaussian density ranking (solid) and the asymptotic Gaussian possibility contour (dashed).}
\label{fig:weibull.pl}
\end{figure} 

While the Weibull model is common in applications, the Weibull model parameters themselves are difficult to interpret.  A relevant and interpretable feature is the mean, which is given by $\Phi = \Theta_2 \, \Gamma(1 + \Theta_1^{-1})$, where $\Gamma$ is the usual gamma function.  Like in the gamma and logistic regression examples above, it is straightforward to obtain an approximate marginal IM for the relevant feature $\Phi$ or, equivalently $\log\Phi$, based on the samples of $(\Theta_1,\Theta_2)$ from $\prior_z^\star$.  Figure~\ref{fig:weibull.mpl}(a) shows a histogram of the samples of $\log\Phi$ and, since there's a slight sign of asymmetry, I use the kernel density estimate-based ranking function to construct the marginal contour for $\log\Phi$ in Figure~\ref{fig:weibull.mpl}(b).  From here, a nominal 90\% confidence interval for $\log\Phi$ can be immediately read off, i.e., $(6.41, 7.74)$.  Exponentiating the endpoints gives a 90\% confidence interval for $\Phi$.


\begin{figure}[t]
\begin{center}
\subfigure[Samples of $\log\Phi$]{\scalebox{0.55}{\includegraphics{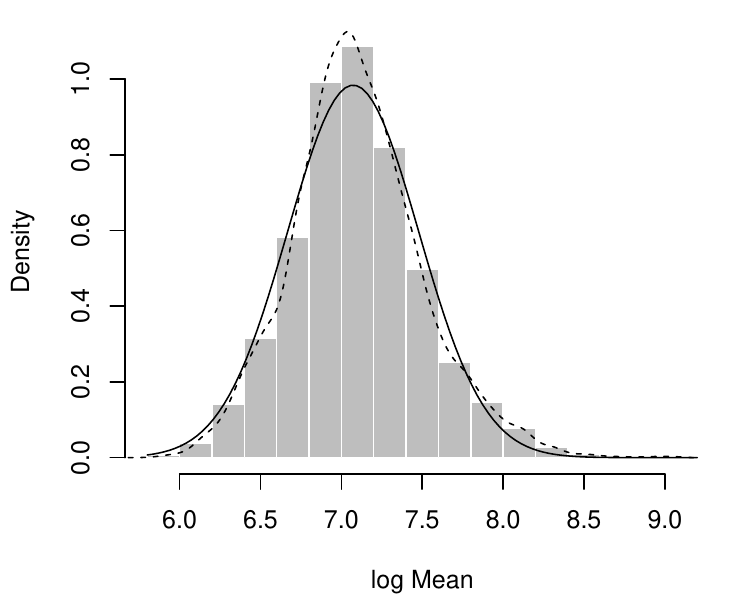}}}
\subfigure[Marginal contour for $\log\Phi$]{\scalebox{0.55}{\includegraphics{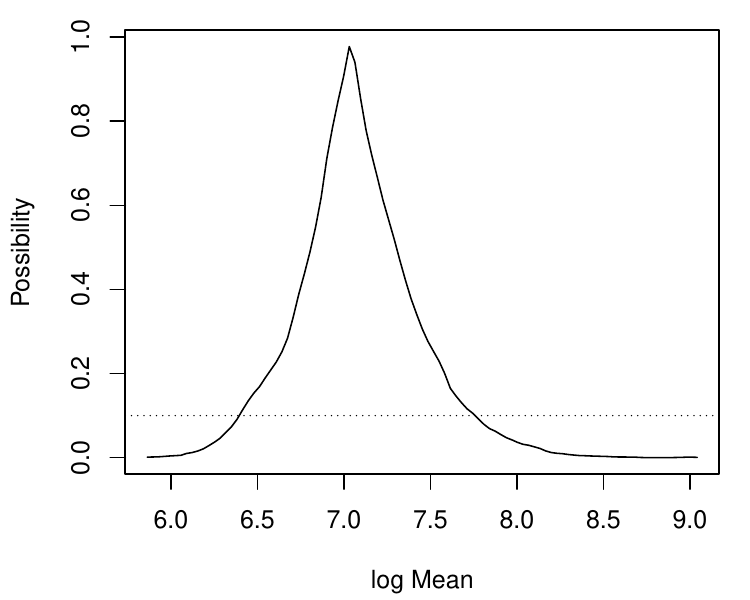}}}
\end{center}
\caption{Results on the mean survival time $\Phi = \Theta_2 \, \Gamma(1 + \Theta_1^{-1})$ of the censored Weibull model.  Panel~(a) shows the distribution of $\log\Phi$ under $\prior_z^\star$ and Panel~(b) shows the corresponding marginal IM contour for $\log\Phi$.}
\label{fig:weibull.mpl}
\end{figure} 
\end{ex}

\section{Simulated-data comparisons}
\label{S:simulations}

Numerous examples in which the IM solution compares favorably to existing methods can be found in the literature; see, e.g., \citet{imvch}, \citet{imunif}, \citet{imconformal.supervised}, and the references cited in Sections~\ref{S:intro}, \ref{S:background}, and \ref{S:examples}.  But the proposed approximation is at least nominally different than the IM solutions that have been previously investigated, so it's of interest to assess how it compares to Bayesian and likelihood-based solutions.  Since what's being proposed here is a general approximation, I want to focus my comparisons on ``standard'' models for which the general proposal is expected to work well off-the-shelf.  For the sake of brevity, I'll focus on those ``standard'' models already introduced in Section~\ref{S:examples}, namely, the gamma and censored Weibull models; I'll revisit the logistic regression example in Section~\ref{S:application}.  Since these are ``standard'' models, the corresponding standard Bayesian and likelihood-based solutions will perform reasonably well, so the (exact or approximate) IM solutions can't perform significantly better.  The take-away message here is that, for these two ``standard'' examples, the proposed approximate IM solutions perform as well---if not slightly better---than the existing Bayesian and likelihood-based solutions, which themselves are very good.  Follow-up papers will consider more complex settings with model-specific improvements to the general approximation proposed here and will conduct more thorough comparisons.  

\begin{ex1}[cont]
For the gamma model, I compare the approximate IM solution to the Bayesian solution with Jeffreys prior and the likelihood-based solution (i.e., first-order asymptotic normality of the maximum likelihood estimator) in terms of coverage probability of the joint credible/confidence sets for $\Theta = (\Theta_1, \Theta_2)$, the gamma shape and scale.  In this simulation study, I take 1000 independent samples of size $n=20$ with true $\Theta_1=7$ and $\Theta_2=3$.  Figure~\ref{fig:sims}(a) shows the coverage probability of the three solutions as a function of the nominal/target confidence level.  All three solutions have actual coverage that nearly matches the target coverage, with both the IM and likelihood-based solutions being a bit closer to the target than the Bayes solution over a range of practical levels.  
\end{ex1}

\begin{ex3}[cont]
For the censored Weibull model, again I compare the approximate IM solution to the Bayesian solution with Jeffreys prior and the first-order likelihood-based solution in terms of coverage probability of credible/confidence sets for $\Theta=(\Theta_1,\Theta_2)$, the Weibull shape and scale.  I take 1000 samples of size $n=20$ from the Weibull distribution with $\Theta_1=2$ and $\Theta_2=3$, but subject to Type~I right censoring with respect to the $\unif(0,4)$ distribution.  Naturally, the censoring creates and additional challenge in that the IM solution described in Example~\ref{ex:censored}---which relies on a plug-in estimate of the censoring distribution---is only asymptotically valid.  So, it's not clear what to expect of the proposed approximation of an only approximately valid IM.  Figure~\ref{fig:sims}(b) shows the coverage probability of the three solutions as a function of the nominal/target confidence level.  Again, all three solutions are high-quality, but my proposed IM approximation has coverage probability at or slightly above the target level across virtually the entire range.  
\end{ex3}

\begin{figure}[t]
\begin{center}
\subfigure[Gamma model]{\scalebox{0.55}{\includegraphics{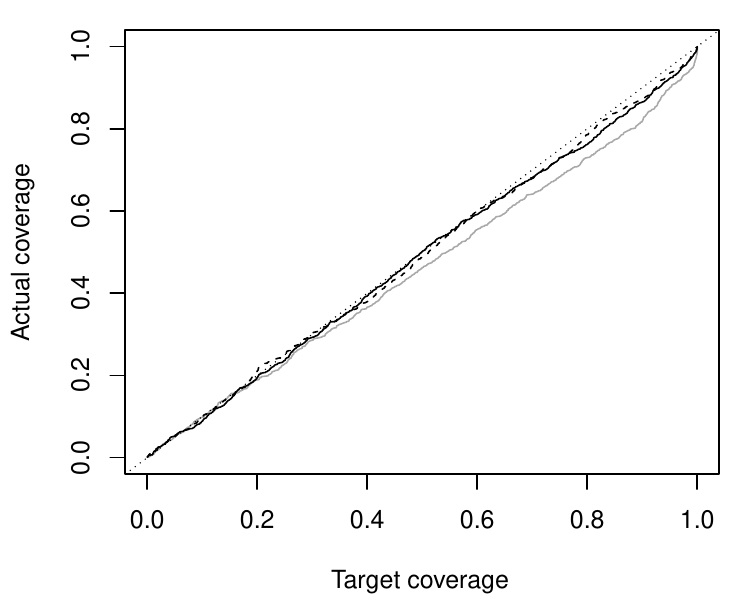}}}
\subfigure[Censored Weibull model]{\scalebox{0.55}{\includegraphics{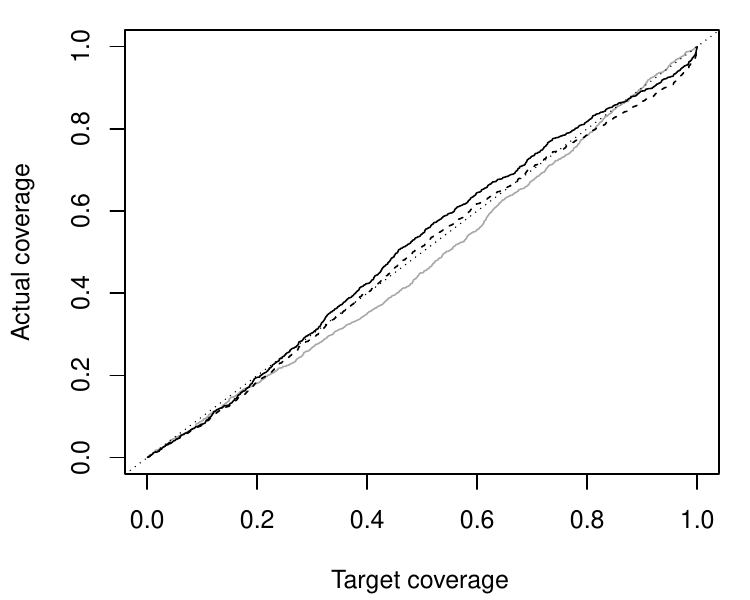}}}
\end{center}
\caption{Actual coverage probability versus the nominal/target level for three different solutions: Jeffreys prior Bayes (gray), first-order likelihood (dashed), and the proposed approximation to the possibilistic IM (solid).  Diagonal dotted line is for reference.}
\label{fig:sims}
\end{figure}

\section{A higher-dimensional illustration}
\label{S:application}

Finally, I present a higher-dimensional version of the logistic regression problem in Example~\ref{ex:logistic}.  Specifically, I consider the Pima Indians diabetes data set, widely used as a benchmark for binary classification, where the goal is to relate incidence of diabetes in Pima Indian women to eight other concomitant variables, namely, number of pregnancies, glucose concentration, blood pressure, skin thickness, insulin, BMI, diabetes pedigree function, and age.  The original data set contains 768 cases.  There are, however, a number of missing values, as indicated in the {\tt PimaIndiansDiabetes2} version \citep[see the R package {\tt mlbench}][]{mlbench}, and, after removing those, I'm left with $n=392$ complete cases.  Compared to the previous examples, this one is relatively high-dimensional: parameter space $\TT$ is of dimension $d=9$.  

The proposed possibilistic IM approximation can easily handle problems of this magnitude, whereas such an example would be completely out of practical reach using the naive computational strategy advocated for in previous IM papers.  A plot summarizing the nine-dimensional inner probabilistic approximation $\prior_z^\star$ is shown in Appendix~\ref{A:numerical} but here, for ease of visualization, Figure~\ref{fig:pima} shows two marginal IM contours: one for the linear predictor $\Phi = \Theta^\top x_\text{new}$ and one for the corresponding diabetes probability $\Psi = \text{logit}^{-1}(\Phi)$, where $x_\text{new}$ is the vector corresponding to the componentwise average of the eight features; note that this summary is based on the indirect marginalization strategy, utilizing the appropriate transformations of the samples of $\Theta$ from the original inner probabilistic approximation $\prior_z^\star$.  The triangles in the two plots denote the likelihood-based asymptotic confidence limits, which are both valid and efficient in relatively large samples like available here.  So, the fact that the IM's confidence limits exactly agree with these is a sign that (a)~the proposed approximation is accurate and (b)~that there's no loss of reliability here despite my use of the simple indirect marginalization strategy.  



\begin{figure}[t]
\begin{center}
\subfigure[Marginal contour for $\Phi=\Theta^\top x_\text{new}$]{\scalebox{0.55}{\includegraphics{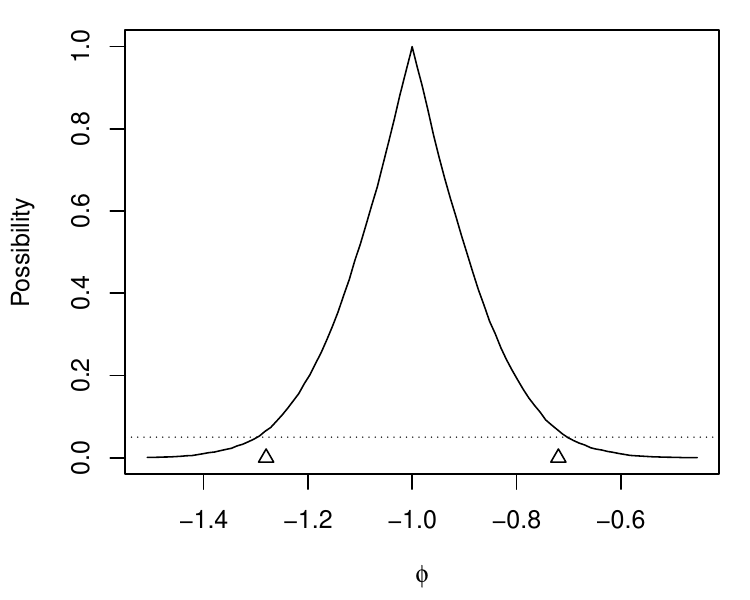}}}
\subfigure[Marginal contour for $\Psi=\text{logit}^{-1}(\Phi)$]{\scalebox{0.55}{\includegraphics{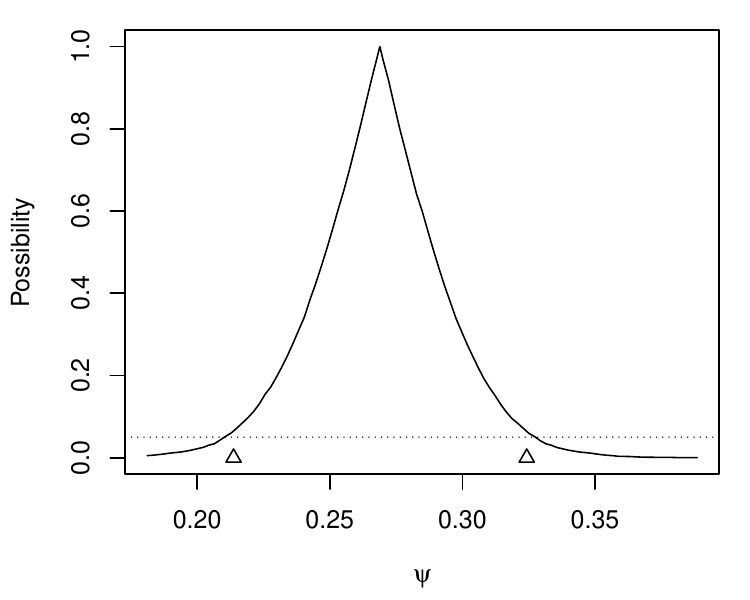}}}
\end{center}
\caption{Marginal contour for the linear predictor $\Phi = \Theta^\top x_\text{new}$ and for the corresponding probability $\Psi = \text{logit}^{-1}(\Phi)$ of a patient with feature vector $x_\text{new}$ having diabetes. Triangles denote the 95\% likelihood-based asymptotic confidence limits.}
\label{fig:pima}
\end{figure}


\section{Conclusion}
\label{S:discuss}

According to \citet{cui.hannig.im}, the IM framework is a fundamental advancement in statistical inference.  To date, these advances have mostly been foundational, theoretical, and methodological, with practically important computational advances unfortunately dragging behind.  The present paper breaks that trend by offering new, simple, and computationally Monte Carlo sampling-driven procedures for approximating various summaries of the IM's possibilistic output.  The non-trivial numerical examples presented here highlight the accuracy of the proposed Monte Carlo approximations which are achieved at just a tiny fraction---seconds versus minutes---of the time required to evaluate the IM output directly using naive, brute-force strategies.  

The key innovation is to stitch together a collection of simple but individually-inadequate approximations into a superior meta-approximation.  Here, for the individual approximations, I'm using the strategy advanced in \citet{imvar.ext}, rather than the bootstrap-based strategy offered in \citet{calibrated.boostrap}, because the former is computationally simpler than the latter.  There may, however, be applications in which the additional flexibility offered by a bootstrap-based strategy is worth the extra computational costs.  This question/topic will be investigated elsewhere.

The results here mark a first and important step in a series of developments leading to an IM toolbox for more-or-less off-the-shelf use by practitioners.  The obvious next step is the development of general-purpose software for those common models---like the ones in the above examples---often encountered in applications.  A natural next step would be extending the methods proposed here in various directions.  This includes extensions to general non- and semiparametric models, where there may not be a likelihood function \citep[e.g.,][]{cella.martin.imrisk}, which is well within reach.  This would also include problems that involve ``model uncertainty,'' i.e., where all or part of what's unknown and to be learned is the underlying model structure.  A good example of this is mixture models, where the number of mixture components is unknown.  Structure learning problems often require regularization, since the data alone cannot rule out overly-complex models, so this leads naturally into extensions of the proposed computational framework to accommodate cases with what \citet{martin.partial, martin.partial2} has referred to as ``partial prior information.''  Finally, while the new proposal here, and the ideas in \citet{imvar.ext} and elsewhere that it's built on, aren't specific to models involving low-dimensional unknowns, there are sure to be challenges associated with scaling up this proposal to handle high-dimensional problems.  My current suggestion is to focus this effort on improving the efficiency and/or flexibility of the variational approximation in Section~\ref{SS:computation}.

\section*{Acknowledgments}

This work is partially supported by the U.S.~National Science Foundation, under grants SES--2051225 and DMS--2412628.  Thanks to the anonymous reviewers for their feedback on an earlier version, and to Emily Hector, Chuanhai Liu, Sahil Patel, and James Robertson for helpful conversations and some  valuable assistance. 


\appendix

\section{Possibility theory primer}
\label{A:primer}

Possibility measures \citep[e.g.,][]{dubois.prade.book, dubois2006} are among the simplest forms of imprecise probability, closely linked to fuzzy set theory \citep[e.g.,][]{zadeh1978} and Dempster--Shafer theory \citep[e.g.,][]{shafer1976, shafer1987}.  Probability and possibility theory differ philosophically and mathematically, but here I'll only briefly discuss the latter.  

The mathematical differences between probability and possibility theory can be succinctly summarized as follows: 
\begin{quote}
{\em Optimization is to possibility theory what integration is to probability theory.}
\end{quote}
That is, a possibility measure $\uPi$ defined on a space $\ZZ$ is determined by a function $\pi: \ZZ \to [0,1]$ with the property that $\sup_{z \in \ZZ} \pi(z) = 1$.  This function is called the {\em possibility contour} and the supremum-equals-1 property is a normalization condition akin to the integral-equals-1 property satisfied by probability densities.  Then, naturally, the possibility measure $\uPi$ is determined by optimizing its contour, i.e., $\uPi(A) = \sup_{z \in A} \pi(z)$, for any $A \subseteq \ZZ$, just like a probability measure is determined by integrating its density.  

The possibility calculus just described is conceptually straightforward and easy to analyze but, of course, there may be computational challenges associated with implementing the theory---just like with probability.  Moreover, this optimization-centric framework of uncertainty quantification is no less rigorous than the familiar probability theory based on (Lebesgue) integration.  To justify this latter claim, I'll make the following two points.  

First, and of particular relevance for the developments in this paper, the aforementioned supremum-equals-1 normalization condition ensures that $\uPi$ is a coherent upper probability \citep[e.g.,][]{cooman.poss1, cooman.aeyels.1999, walley1997} in the spirit of \citet{walley1991} and others.  Among other things, this means that $\uPi$ determines a non-empty (closed and convex) set of ordinary probabilities that it dominates:
\begin{equation}
\label{eq:credal}
\cred(\uPi) = \{ \prior \in \text{probs}(\ZZ): \prior(H) \leq \uPi(H) \text{ for all measurable $H$}\}, 
\end{equation}
where $\text{probs}(\ZZ)$ is the set of probabilities supported on the Borel $\sigma$-algebra of measurable subsets in $\ZZ$. The set $\cred(\uPi)$ is called the {\em credal set} and all (coherent) upper probabilities have one.  Aside from being relatively simple, a key advantage to possibilistic uncertainty quantification is that the associated credal set has a statistically oriented characterization, which I discuss in the main paper. 

Second, one can extend possibilistic evaluations beyond just the possibility values assigned to hypotheses.  This extension operation is analogous to the evaluation of expected values via Lebesgue integration of measurable functions.  This extension relies on what's called {\em Choquet integration}, originally developed in \citet{choquet1953}, thoroughly analyzed in  \citet{lower.previsions.book} and elsewhere, with applications in statistics \citep[e.g.,]{huber1973.capacity}.  Note that Choquet integration is not specifically tied to possibility theory, it's for integration with respect to general non-additive capacities, and possibility measures are a special case; there are, however, some nice simplifications that happen because of the {\em maxitive}---as opposed to additive---form of possibility measures.  Suppose that $h: \TT \to \RR$ is a non-negative function; non-negativity is just for convenience here.  Then the {\em Choquet integral} of $h$ with respect to the possibility measure $\uPi$ is given by
\begin{equation}
\label{eq:choquet}
\uPi \, h := \int_{\inf h}^{\sup h} \Bigl\{ \sup_{\theta \in \TT: h(\theta) > s} \pi(\theta) \Bigr\} \, ds, 
\end{equation}
where $\pi$ is the contour function associated with $\uPi$.  As an important special case, connecting the possibilistic marginalization strategies discussed in Appendix~\ref{A:marginal}, suppose that $h(\theta) = 1(\theta \in H)$, where $H = \{\theta: k(\theta) \in K\}$ for some function $k$ and some subset $K \subseteq k(\TT)$.  Then applying the Choquet integral formula in \eqref{eq:choquet} gives 
\[ \uPi \, h = \sup_{\theta: h(\theta)=1} \pi(\theta) = \sup_{\theta \in H} \pi(\theta) = \sup_{\kappa \in K} \sup_{\theta: k(\theta)=\kappa} \pi(\theta). \]
There's a clear analogy that can be made between the right-hand side of the above display and the more familiar Bayesian calculus: the degree of possibility/probability assigned to ``$\Theta \in H$'' or, equivalently, to ``$k(\Theta) \in K$'' is obtained by first getting the marginal contour/density for $k(\Theta)$ at $\kappa$ by optimizing/integrating over $\{\theta: k(\theta)=\kappa\}$ and then optimizing/integrating over $\kappa \in K$.  The same can be said for Choquet integrals of other kinds of functions $h$, but the point is most easily made when $h$ is an indicator.

\section{Proof of Theorem~\ref{thm:credal.char} and Corollary~\ref{cor:credal.char}}
\label{A:proof}

To prove sufficiency, according to \eqref{eq:credal.1}, it's enough to check that, if $\prior$ is as defined in \eqref{eq:credal.2}, then $\prior(C_\alpha) \geq 1-\alpha$ for each $\alpha$.  Using the fact that the $\alpha$-cuts are nested, this follows from the simple manipulation, 
\begin{align*}
\prior(C_\alpha) & = \int_0^1 \kernel^\beta( C_\alpha ) \, \marg(d\beta) \\
& = \int_0^\alpha \underbrace{\kernel^\beta(C_\alpha)}_{\geq 0} \, \marg(d\beta) + \int_\alpha^1 \underbrace{\kernel^\beta(C_\beta)}_{= 1} \, \marg(d\beta) \\
& \geq \marg([\alpha,1]) \\
& \geq 1-\alpha, 
\end{align*}
where the last inequality follows by the stochastically-no-smaller-than-$\unif(0,1)$ property of $\marg$.  To prove necessity,
take the given $\prior \in \cred(\uPi)$ and consider a random element $\Theta \sim \prior$.  Note that $\Theta \in \partial C_{\pi(\Theta)}$ with $\prior$-probability~1.  For some intuition, imagine partitioning $\TT$ based on these level sets; then $\Theta$ itself is determined by the level set it's on together with its position on the level set, and hence $\prior$ corresponds to an average of the conditional distribution of $\Theta$, given $\pi(\Theta)$, with respect to the marginal distribution of $\pi(\Theta)$.  It follows from \eqref{eq:credal.1} that the random variable $\pi(\Theta)$ is stochastically no smaller than $\unif(0,1)$, and let $\marg$ be its marginal distribution relative to $\prior$.  Similarly, take $\kernel^\beta$ to be (a version of) the conditional distribution of $\Theta$, given $\pi(\Theta)=\beta$, relative to $\prior$---note that $\kernel^\beta$ is fully supported on $\partial C_\beta \subset C_\beta$, as required.  Then the equality \eqref{eq:credal.2} follows from the law of iterated expectation, completing the proof of Theorem~\ref{thm:credal.char}.

Equality throughout the above display is achieved if, first, $\kernel^\beta(C_\alpha)=0$ for all $\alpha < \beta$ and, second, if $\marg$ is $\unif(0,1)$.  This condition on the kernel is achieved if $\kernel^\beta$ is supported on the boundary, $\partial C_\beta$, of $C_\beta$, for each $\beta \in [0,1]$, hence Corollary~\ref{cor:credal.char}.

\section{Direct implementation of Corollary~\ref{cor:credal.char}}
\label{A:direct}

Corollary~\ref{cor:credal.char} describes the inner probabilistic approximation of the IM $\uPi_z$ as a mixture of kernels $\kernel_z^\alpha$ supported on the boundary $\partial C_\alpha(z)$ of the possibility contour's $\alpha$-cut, for each $\alpha \in [0,1]$.  The reader might mistakingly think that, because $\kernel_z^\alpha$ is supported on a lower-dimensional space, handling such a mixture is an impractically impossible task.  The cause of this confusion is that the reader imagines having to {\em condition} on the probability-0 event associated with the lower-dimensional boundary.  While conditioning on the boundary event is one way to construct such a kernel, that's not the only way.  Instead, I can just define $\kernel_z^\alpha$ to be a distribution supported on $\partial C_\alpha(z)$, and there would rarely be problems with carrying this out.  For instance, an incredibly mild assumption is that each $C_\alpha(z)$ be star-convex, meaning that there exists a point on the interior (e.g., the maximum likelihood estimator $\hat\theta_z$) such that every point on the boundary can be uniquely identified by a ray emanating from that interior point.  The cartoon illustration in Figure~\ref{fig:bean} in the main text is star-convex but not elliptical/convex.  With this structure, sampling on the boundary $\partial C_\alpha(z)$ is equivalent to sampling on the unit ball, which is not particularly challenging.  Then one just needs to identify points on the unit ball with their corresponding points on $\partial C_\alpha(z)$, which would generally require some iterative numerical methods; see below.  An advantage of the second-layer approximation with an ellipsoidal $C_\alpha^\out(z)$ is that the mapping from the ball to $\partial C_\alpha^\out(z)$ is explicit.

The main text made the claim that the elliptically-shaped contour formulation is not crucial the development; such an assumption is imposed solely because (a)~it's the ``right'' choice for large-samples so there's little room for improvement, (b)~it makes the computations relatively simple, and (c)~there is no empirical evidence indicating that this approximation is inaccurate.  I make no claims, however, that the elliptical formulation is universal.  Fortunately, it's easy to extend the method to handle other forms, even the inner probabilistic approximation exactly.  Below I'll describe implementation of the inner probabilistic approximation itself since that's arguably the most general and the most complicated.  Other problem-specific approximations could be considered too, e.g., with $C_\alpha^\sigma(z)$ a suitable Kullback--Leibler neighborhood when the parameter space is the simplex, but these details are beyond the scope of the present paper.  

For this discussion, suppose that the level sets of $\pi_z$ are non-elliptical (like in Figure~\ref{fig:bean} in the main text) but are star-convex in the sense described above.  This is a very mild supposition.  Then to sample on the boundary $C_\alpha(z)$, I just need to sample on the unit ball, i.e., sample a random direction $U$ in $d$-space, and then follow the line in that direction until it intersects with $\partial C_\alpha(z)$.  That is, solve the equation $\pi_z(\hat\theta_z + t U) = \alpha$ as a function of $t > 0$.  Since $\pi_z$ can only be evaluated using Monte Carlo, this root-solving problem can't be solved using Newton's method or bisection---it requires something like the stochastic approximation algorithm employed by Cella and Martin as discussed in Appendix~\ref{A:imvar}.  To summarize, independent samples from the inner probabilistic approximation $\prior_z^\inn$ are drawn by iterating the following two steps:
\begin{enumerate}
\item Draw $A \sim \unif(0,1)$.
\item Given $A=\alpha$, draw $U$ uniformly on the unit $d$-dimensional ball (e.g., by drawing a $d$-dimensional standard Gaussian and then normalizing it) and then set $\Theta$ equal to the root of the function $g(t) = \pi_z(\hat\theta_z + tU) - \alpha$, $t > 0$.  
\end{enumerate}
While carrying out the above computations are not difficult, they're more expensive than the solution advocated for in the main text based on the elliptical approximation.  At least part of the reason for this increased expense is that one needs to perform the root-solving $M$ many times, where $M$ is the desired Monte Carlo sample size, typically in the thousands, whereas the strategy described in the paper only needs root-solving to find $\sigma(z,\alpha)$ on a fixed set of about 100 different $\alpha$ values; see Remark~\ref{re:parallel} below.  

Again, no claims are made that the elliptical approximation is universal, but it's simple, asymptotically optimal, and works well in a wide range of examples, including those with relatively small sample sizes.  Future investigations will explore IM solutions to more complex problems where the full generality of the proposed computational strategy, as described above, will surely be needed.

\section{Cella and Martin's algorithm}
\label{A:imvar}

The goal here is to review some of the relevant details presented in \citet{imvar.ext} concerning the choice of $\sigma=\sigma(z,\alpha)$ employed in Section~\ref{SS:computation} in the main text.  The focus here is on the case where the $\alpha$-cuts $C_\alpha^\sigma(z)$ are ellipsoidal, but similar things can be done with other shapes, as discussed in Appendix~\ref{A:direct} above.  

Recall the eigenvalue-adjusted version $J_z(\sigma)$ of the observed Fisher information matrix $J_z$, i.e., based on the spectral decomposition $J_z = E \Lambda E^\top$.  Also recall the Gaussian possibility's $\alpha$-cut in \eqref{eq:var.cut}:
\[ C_\alpha^\sigma(z) = \{\theta: (\theta - \hat\theta_z)^\top \, J_z(\sigma) \, (\theta - \hat\theta_z) \leq F_d^{-1}(1-\alpha)\}, \quad \alpha \in [0,1]. \]
The key observation is that the desired containment $C_\alpha^\sigma(z) \supseteq C_\alpha(z)$ holds if and only if 
\[ \sup_{\theta \not\in C_\alpha^\xi(z)} \pi_z(\theta) \leq \alpha. \]
Since the contour $\pi_z$ is itself approximately Gaussian \citep{imbvm.ext} and the maximum likelihood estimator $\hat\theta_z$, also the mode of $\pi_z$, is in $C_\alpha^\sigma(z)$, it's expected that the above supremum is attained on the boundary $\partial C_\alpha^\sigma(z)$.  Moreover, since equality in the above display implies a near-perfect match between the IM's and proposed elliptical $\alpha$-cuts, a reasonable goal is to find a root to the function 
\[ g_\alpha(\sigma) := \max_{\theta \in \partial C_\alpha^\sigma(z)} \pi_z(\theta) - \alpha. \]
Design of an iterative algorithm to find this root requires care, primarily because evaluating $\pi_z$ is expensive; so the goal is to evaluate $g_\alpha(\sigma)$ with as few $\pi_z$ evaluations as possible.  \citet{imvar.ext} propose to represent the boundary of $C_\alpha^\sigma(z)$ by $2d$-many vectors 
\begin{equation}
\label{eq:posts}
\vartheta_s^{\sigma, \pm} := \hat\theta_z \pm \bigl\{ F_d^{-1}(1-\alpha) \, \sigma_s \, / \, \lambda_s \}^{1/2} \, e_s, \quad s=1,\ldots,d, 
\end{equation}
where $(\lambda_s, e_s)$ is the eigenvalue--eigenvector pair corresponding to the $s^\text{th}$ largest eigenvalue in the spectral decomposition of $J_z$ mentioned above.  Then define the vector-valued function $\hat g_\alpha$ with components 
\begin{equation}
\label{eq:ghat.xi}
\hat g_{\alpha,s}(\sigma) = \max\{ \pi_z(\vartheta_s^{\sigma, +}), \pi_z(\vartheta_s^{\sigma,-})\} - \alpha, \quad s=1,\ldots,d.
\end{equation}
At least intuitively, negative and positive $\hat g_{\alpha,s}(\sigma)$ indicate that the $\alpha$-cut $C_\alpha^\sigma(z)$ is too large and too small, respectively in the $e_s$-direction.  From here, \citet{imvar.ext} suggest applying a stochastic approximation algorithm \`a la \citet{robbinsmonro} and \citet{kushner} to construct a sequence $(\sigma^{(t)}: t \geq 1)$ of $d$-vectors that converges to a root of $\hat g_\alpha$ and, hence, an approximate root of $g_\alpha$.  For an initial guess $\sigma^{(0)}$, the specific sequence is defined as 
\[ \sigma_s^{(t+1)} = \sigma_s^{(t)} + w_{t+1} \, \hat g_{\alpha,s}(\sigma^{(t)}), \quad s=1,\ldots,d, \quad t \geq 0, \]
where $(w_t)$ is a deterministic sequence that satisfies 
\[ \sum_{t=1}^\infty w_t = \infty \quad \text{and} \quad \sum_{t=1}^\infty w_t^2 < \infty. \]
These steps are iterated until (practical) convergence is achieved, and the limit is what I called $\sigma(z,\alpha)$ in Section~\ref{SS:computation}.  The basic steps are outlined in Algorithm~\ref{algo:varim}, but I refer to \citet{imvar.ext} for further details and discussion.  

\begin{algorithm}[t]
\SetAlgoLined
requires: data $z$, eigen-pairs $(\lambda_s, e_s)$, and ability to evaluate $\pi_z$\; 
initialize: $\alpha$-level, guess $\sigma^{(0)}$, step size sequence $(w_t)$, and threshold $\eps > 0$\; 
set: {\tt stop = FALSE}, $t=0$\; 
\While{{\tt !stop}}{
construct the representative points $\{\vartheta_s^{\sigma^{(t)}, \pm}: s=1,\ldots,d\}$ as in \eqref{eq:posts}\; 
evaluate $\hat g_{\alpha,s}(\sigma_s^{(t)})$ for $s=1,\ldots,d$ as in \eqref{eq:ghat.xi}\;
update $\sigma_s^{(t+1)} = \sigma_s^{(t)} \pm w_{t+1} \, \hat{g}_{\alpha,s}(\sigma_s^{(t)})$ for $s=1,\ldots,d$\;
\eIf{$\max_s |\sigma_s^{(t+1)} - \sigma_s^{(t)}| < \eps$}{
  $\sigma(z,\alpha) = \sigma^{(t+1)}$\;
  {\tt stop = TRUE}\;
}{
  $t \gets t+1$\;
}
}
return $\sigma(z,\alpha)$\; 
\caption{Determining $\sigma(z,\alpha)$---from \citet{imvar.ext}.}
\label{algo:varim}
\end{algorithm}

\section{Miscellaneous technical remarks}
\label{A:remarks}

\begin{remark}
\label{re:parallel}
To briefly follow up on the parallelization point in Section~\ref{SS:computation}, there's practical adjustment to the above procedure that dramatically speeds up computations.  Start with a fixed grid $\mathcal{A}$ of, say 100 $\alpha$ values from 0.001 to 0.999.  Then evaluate $\sigma(z,\alpha)$ for all $\alpha \in \mathcal{A}$, which sets a fixed limit on the amount of computational investment required.  Now, for each $A \sim \unif(0,1)$ drawn in Step~1 of Algorithm~\ref{algo:inner}, find the two $\alpha$ values in $\mathcal{A}$ that sandwich $A$ and find $\sigma(z,A)$ via linear interpolation; then proceed to Step~2.  Once the 100 values of $\sigma(z,\cdot)$ are obtained, the samples from $\prior_z^\star$ are virtually free.
\end{remark}

\begin{remark}
\label{re:asymptotics}
Here I offer a brief large-sample analysis to justify the claimed accuracy of the probabilistic approximation.  Without loss of persuasiveness, I give only a heuristic argument building on the rigorous analysis in \citet{imbvm.ext}.  

As mentioned above, \citet{imbvm.ext} show that, under the regularity conditions sufficient for asymptotic normality and efficiency of the maximum likelihood estimator, the relative likelihood-based possibilistic IM enjoys a large-sample Gaussianity property, akin to the classical Bernstein--von Mises theorem fundamental to Bayesian analysis; see \eqref{eq:approx.gaussian}.  The uniform mode of their convergence result implies, first, that the IM's $\alpha$-cuts, $C_\alpha(z)$, will merge with the $\alpha$-cuts, $C_\alpha^1(z)$, of the elliptical approximation, with $\sigma=1$, as the size of the sample $z$ increases to $\infty$.  Second, since $\sigma(z,\alpha)$ is designed to make the $\alpha$-cuts $C_\alpha^{\sigma(z,\alpha)}(z)$ agree with the original IM's, one fully expects that the components of $\sigma(z,\alpha)$ are converging to 1 as the sample size increases.  

Putting everything together, the original IM's $\alpha$-cuts merge with the corresponding Gaussian $\alpha$-cuts asymptotically, and the components of the index $\sigma(z,\alpha)$ identified by Cella and Martin's algorithm are converging to 1.  Therefore, the symmetric difference $C_\alpha(z) \, \triangle \, C_\alpha^{\sigma(z,\alpha)}(z)$ is converging to $\varnothing$, which implies that the proposed sampling algorithm exactly recovers the IM's limiting inner probabilistic approximation, which is the Gaussian distribution just like in the example in Section~\ref{SS:insights}.
\end{remark}

\begin{remark}
\label{re:caliboot}
As mentioned briefly in the main text (see Section~\ref{SS:computation}), the proposed IM approximation has some connections to and, in a certain sense, is inspired by the {\em calibrated bootstrap} developments in \citet{calibrated.boostrap}.  I won't attempt to give a complete description of the calibrated bootstrap approach here; I'll only give enough detail---using my own context and notation---that a meaningful comparison with my proposed solution can be made.  While the authors don't describe their approach in this way, what they are proposing relies (implicitly) on a characterization result like in Theorem~\ref{thm:credal.char} and their implementation follows the same basic form as in Algorithm~\ref{algo:inner}.  Where the two approaches differ is, of course, in the details behind Step~2 of Algorithm~\ref{algo:inner}.  Where I construct a mixture of kernels supported on suitable $\alpha$-cut boundaries, \citet{calibrated.boostrap} adopt a bootstrap strategy.  My construction of the kernels requires certain $\alpha$-specific tuning and so does their bootstrap---they work with a $m$-out-of-$n$ bootstrap strategy where the ``$m$'' is suitably tuned to a particular $\alpha$ level.  They then fix a range of $\alpha$s and evaluate $m=m(z,\alpha)$ at each $\alpha$ on that range; compare this to my modification of Algorithm~\ref{algo:inner} described in Remark~\ref{re:parallel} above.  For each $m=m(z,\alpha)$ values, determined by the range of $\alpha$s, a bunch of $m$-out-of-$n$ bootstrap samples are taken, resulting in a collection of $\alpha$-specific bootstrap distributions of the maximum likelihood estimator.  All these samples from $\TT$, call them $\Theta^{(s)}$, along with their corresponding $\pi_z(\Theta^{(s)})$ values, are aggregated into a bag, say, $\{\Theta^{(s)}, \pi_z(\Theta^{(s)})\}$, indexed by $s$.  Finally, they repeatedly sample $A \sim \unif(0,1)$ and, for each $A$, they identify the $\Theta^{(s)}$ such that $|\pi_z(\Theta^{(s)}) - A|$ is minimized.  The empirical distribution of those $\Theta^{(s)}$ values retained in this process can then be compared to the samples from $\prior_z^\star$ that I get from applying my proposed sampling algorithm.  

An advantage of the calibrated bootstrap strategy, compared to the proposal in the present paper, is that it doesn't rely on any assumptions about what roughly the form of the $\alpha$-cuts of $\pi_z$ would look like; in that sense, calibrated bootstrap is more robust then my solution.  Moreover, at least intuitively, if one goes to the asymptotic limit, each of the $\alpha$-specific bootstrap distributions should be approximately the same and equivalent to the asymptotically Gaussian limiting distribution of the maximum likelihood estimator.  In that case, retaining a draw with contour closest to $A \sim \unif(0,1)$ is roughly like conditioning the aforementioned Gaussian distribution on the boundary of the $A$-specific ellipse---basically the same as what I'm proposing here.  

The downside to calibrated bootstrap, and what motivated my proposal here, is that {\em if} the level sets are roughly elliptical---as they often are in applications, at least with under a suitable parametrization---then there is substantial computational simplifications that can be enjoyed.  In particular, while I can obtain thousands of samples from $\prior_z^\star$ with only a few hundred Monte Carlo evaluations of $\pi_z$, the calibrated bootstrap needs to evaluate $\pi_z$ at every $\Theta^{(s)}$ sampled from the different $m$-out-of-$n$ bootstraps, and enough of these are needed to ensure the range $[0,1]$ for $\alpha$ is adequately covered, so thousands of Monte Carlo evaluations of $\pi_z$ would be needed.  

As I see it, both my proposal and calibrated bootstrap stand on their own---neither uniformly dominates the other.  It's perfectly natural that, for example, in high-dimensional cases where the elliptical form may not be appropriate, something more robust like calibrated bootstrap would be desired.  But in the present paper's context, i.e., the not-very-high-dimensional cases common in everyday applications, the elliptical form (or perhaps some other specific form) could be justified and, in such cases, the ``nonparametric'' aspect of calibrated bootstrap could be considered as overkill, so it's advantageous save computational resources by using a procedure like mine.  
\end{remark}

\section{Details about IM marginalization}
\label{A:marginal}

As briefly mentioned in the main text, there are two basic strategies for eliminating nuisance parameters in the possibilistic IM framework---namely, an {\em extension-based} and a {\em profiling-based} strategy---and I'll discuss each of these briefly here.  At a high level, the extension-based marginalization strategy is conceptually and computationally simple but less statistically efficient, whereas the profiling-based strategy is more complicated but tends to be more efficient.  Importantly, both strategies guarantee that the corresponding marginal IM for the quantities of interest is valid so, unlike Bayes, fiducial, and other attempts at probabilistic inference, there's no risk of losing this basic reliability property in the possibilistic IM framework.  

As indicated in Appendix~\ref{A:primer} and in the main text, to a less extent, optimization is for possibility theory what integration is for probability theory.  So, it makes sense that one could eliminate nuisance parameters from a possibilistic IM by using optimization.  Formally, this operation is known as {\em extension} in the possibility theory literature, dating back at least to \citet{zadeh1978} and it says the following: if $\theta \mapsto \pi_z(\theta)$ is a possibility contour quantifying uncertainty about $\Theta$, then a corresponding possibility contour that quantifies uncertainty about $\Phi = f(\Theta)$ is given by 
\[ \pi_z^\text{\sc ex}(\phi) = \sup_{\theta: f(\theta)=\phi} \pi_z(\theta), \quad \phi \in f(\TT). \]
This extension-based marginalization strategy can be readily compared with the corresponding operation in probability theory: replace the contour by a probability density function and optimization by integration.  This is conceptually and computationally straightforward, since it only involves operating on the already-available contour $\pi_z$.  As stated in the main text, the possibility calculus guarantees the reliability of $\pi_z$ is transferred to its extension $\pi_z^\text{\sc ex}$, regardless of the feature mapping $f$; the reader is encouraged to do the quick 1--2 line proof to verify this claim.  Unfortunately, as is often the case with simple, universal solutions, it need not be the most efficient solution available.  This inefficiency is easy to see empirically in concrete examples, and was proved in general for large samples in \citet{imbvm.ext}.  

A second approach to marginalization takes the specific feature mapping $f$ into consideration in the IM construction.  If $\Phi=f(\Theta)$ is the relevant quantity of interest, then one can define a relative profile likelihood 
\[ R^\text{\sc pr}(z, \phi) = \sup_{\theta: f(\theta)=\phi} R(z,\theta) = \frac{\sup_{\theta: f(\theta)=\phi} L_z(\theta)}{\sup_\theta L_z(\theta)}, \quad z \in \ZZ, \quad \phi \in f(\TT). \]
Then this can be treated as the ranking function in the IM construction, leading to a new, $f$-specific marginal IM contour 
\[ \pi_z^\text{\sc pr}(\phi) = \sup_{\theta: f(\theta)=\phi} \prob_\theta\{ R^\text{\sc pr}(Z,\phi) \leq R^\text{\sc pr}(z,\phi)\}, \quad \phi \in f(\TT). \]
The supremum outside the $\prob_\theta$-probability evaluation above is present because, even though the relative profile like $R^\text{\sc pr}(Z,\phi)$ doesn't directly depend on other aspects of $\theta$, it's distribution generally depends on all of $\theta$; there are some cases in which the relative profile likelihood is a pivot in the sense that its distribution only depends on $\theta$ through $f(\theta)$ and, in such cases, the supremum can be removed.  For obvious reasons, this is referred to as a profile-based marginalization strategy.  It's more involved because it requires the direct construction of an $f$-specific IM, but the payoff for this extra effort is improved efficiency.  Again, this efficiency is easy to see empirically in concrete examples, and it was proved in general for large samples in \citet{imbvm.ext}.  

As it relates to the developments in the present paper, what I referred to as the {\em indirect marginalization approach} using the inner probabilistic approximation would, in ideal cases, correspond to the extension-based marginalization strategy described above.  That is, in certain cases, if $\Theta$ is drawn from the inner probabilistic approximation $\prior_z^\star$ of $\uPi_z$, then the marginal distribution of $f(\Theta)$ derived from that of $\Theta$ using the rules of ordinary probability would, in such cases, correspond to the inner probabilistic approximation of $\uPi_z^\text{\sc ex}$.  The {\em direct marginalization approach}, however, would first directly construct a marginal IM using, say, the profile-based strategy described above, and then obtain the inner probabilistic approximation thereof.  This latter approach is safer and generally more efficient, but the trade-off is that it's not as computationally convenient.

\section{Additional numerical results}
\label{A:numerical}

\subsection{Bivariate normal correlation}

A challenging one-parameter inference problem is the bivariate normal with known means and variances but unknown correlation.  Suppose that $Z=(Z_1,\ldots,Z_n)$ are iid, with $Z_i=(X_i,Y_i)$ a bivariate normal random vector with zero means, unit standard deviations, and correlation $\Theta \in \TT = (-1,1)$ to be inferred.  Interestingly, that the means and variances are known makes the problem more difficult---it's a curved exponential family so the minimal sufficient statistic is not complete and there are various ancillary statistics available to condition on \citep[e.g.,][]{basu1964}.  How this affects asymptotic inference is detailed in \citet{reid2003}.  It is straightforward to construct a possibilistic IM \citep[][Example~2]{martin.basu}, which is exactly valid for all sample sizes and asymptotically efficient, but computation of the naive approximation \eqref{eq:pi.naive} is relatively expensive because there's no pivotal structure and no closed-form expression for the maximum likelihood estimator.  The variational approximation proposed in \citet{imvar.ext} is fast and easy, but it loses the exact validity of the original IM.  Here I apply the proposed Monte Carlo sampling strategy on Fisher's transformation scale, i.e., on $\Psi = \text{arctanh}(\Theta)$, the inverse hyperbolic tangent.  Figure~\ref{fig:cor}(a) shows a histogram of 5000 samples from the distribution $\prior_z^\star$ of $\Psi$ as described above, based on (a centered and scaled version of) the law school admissions data analyzed in \citet{efron1982}, where $n=15$ and the maximum likelihood estimator of $\Theta$ is $\hat\theta_z = 0.789$.  Overlaid on this plot is, first, a normal density with mean $\text{arctanh}(0.789) = 1.07$ and estimated standard deviation and, second, the kernel density estimate; the difference between these two estimates is negligible.  Panel~(b) shows the exact IM contour and three stitched IM approximations as in \eqref{eq:p2p} corresponding to three different ranking functions $r_z$.  All three approximations have roughly the same shape as the true contour, as expected.  The two based on density estimates---Gaussian and kernel---are quite similar and closely agree with the exact contour.  The likelihood-based approximation is very accurate on the right-hand side but a bit conservative on the left-hand side; this is because, as in Figure~\ref{fig:cor}(a), the distribution $\prior_z^\star$ of the Fisher-transformed $\Theta$ is symmetric but the likelihood-based ranking is not.  

\begin{figure}[t]
\begin{center}
\subfigure[Samples from $\prior_z^\star$]{\scalebox{0.55}{\includegraphics{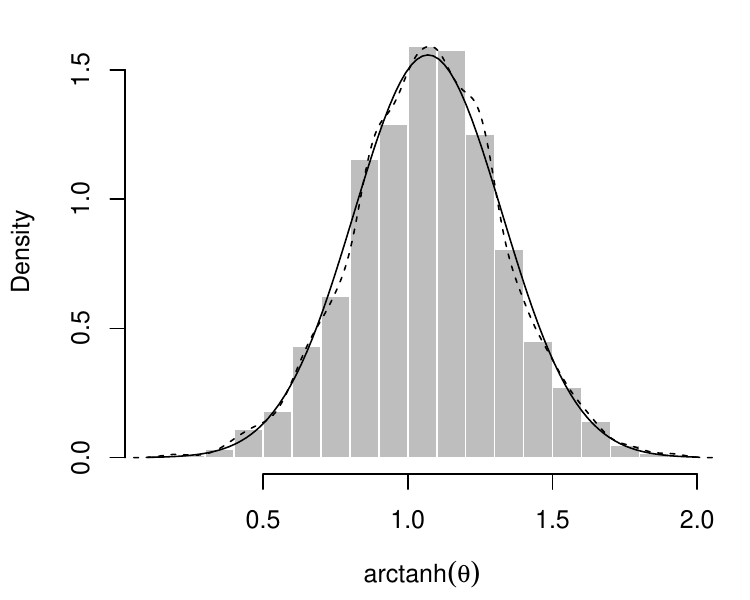}}}
\subfigure[Possibility contours]{\scalebox{0.55}{\includegraphics{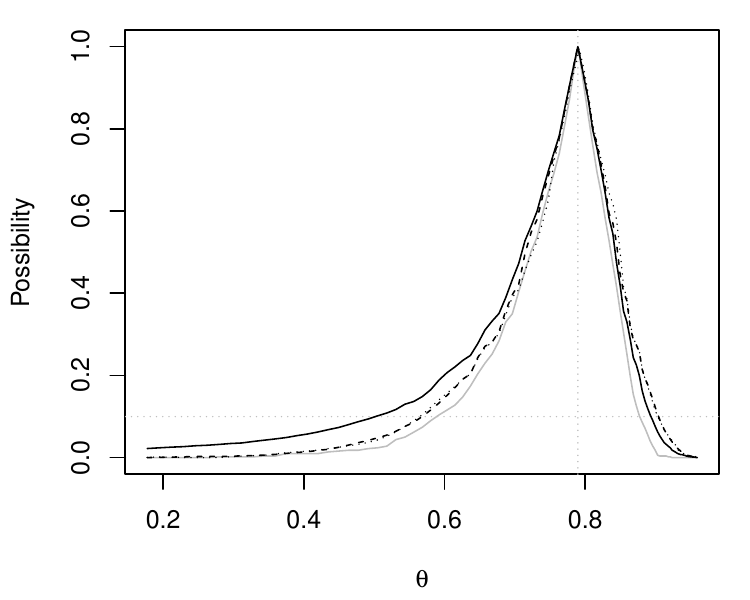}}}
\end{center}
\caption{Summary of the results for the bivariate normal correlation. In Panel~(a), the solid line is the Gaussian density and dashed line is the kernel density estimate from R's {\tt density}.  In Panel~(b), the gray line is the exact IM contour, the dashed and dotted lines are approximations using the Gaussian and kernel density rankings, respectively, and the solid line is the approximation based on the likelihood ranking.}
\label{fig:cor}
\end{figure}

\subsection{Moderate-dimensional logistic regression}

There's not much to say here, just reporting some summary results from fitting the moderate-dimensional logistic regression using the proposed approximate IM strategy.  In this case, the inner probabilistic approximation $\prior_z^\star$ to the possibilistic IM is a nine-dimensional probability distribution.  Figure~\ref{fig:pima.pairs} below partially summarizes the full joint distribution, showing the individual marginal distributions and each of the pairwise joint distributions.  Note that all of these look at least approximately Gaussian, which justifies my choice to construct an approximation of the full, nine-dimensional possibility contour $\pi_z$---and the corresponding marginal possibility contours displayed in Figure~\ref{fig:pima} in the main text---using a Gaussian ranking function as described in Section~\ref{SS:back.to.poss}.  

\begin{figure}[t]
\begin{center}
\scalebox{1.1}{\includegraphics{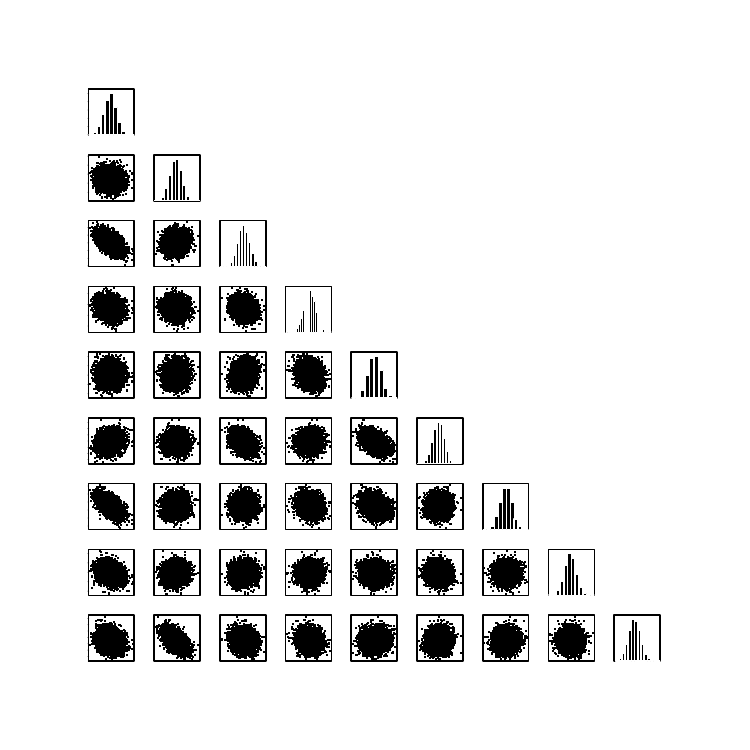}}
\end{center}
\caption{Visual summaries of the marginal and pairwise joint distributions corresponding to the nine-dimenaional inner probabilistic approximation $\prior_z^\star$ in the logistic regression model fit to the Pima Indians diabetes data example.}
\label{fig:pima.pairs}
\end{figure}

\bibliographystyle{apalike}
\bibliography{/Users/rgmarti3/Dropbox/Research/mybib.bib}

\end{document}